\documentclass[aps,prd,twocolumn,superscriptaddress,nofootinbib,floatfix]{revtex4-2}
\usepackage{amsmath,amssymb,bm}
\usepackage{graphicx}
\usepackage{hyperref}
\usepackage{orcidlink}
\usepackage{booktabs}

\newtheorem{theorem}{Theorem}
\newtheorem{proposition}{Proposition}

\newcommand{\evP}{\textnormal{\textsc{[proven]}}}
\newcommand{\evM}{\textnormal{\textsc{[machine-checked]}}}
\newcommand{\evX}{\textnormal{\textsc{[measured]}}}

\begin{document}

\title{Fermionic quantum field theories as probabilistic cellular automata
in three dimensions}

\author{Piotr S. Topa\,\orcidlink{0009-0005-2870-2463}}
\affiliation{Independent researcher}

\date{\today}

\begin{abstract}
Wetterich has shown that certain probabilistic cellular automata in
$1+1$ dimensions are exactly equivalent to fermionic quantum field
theories, and posed the three-dimensional construction as an open
problem. We construct one: a layered automaton (the ``coincube'') whose
conversion blocks form the real quaternion representation of
$\mathrm{Cl}(0,3)$, controlled by binary environment fields whose
streaming outruns the carrier (``fresh tape'')---no bit is ever
consulted twice, and the quenched single-particle ensemble propagator
equals the annealed Bloch operator identically at every time, a theorem
of the model. Exact consequences: the step operator is unitary with
respect to an explicit complex structure, with the spectral analysis
its faithful complexification; the local Grassmann action is extracted
mechanically; the single-carrier excitation is a Weyl fermion with an
exactly isotropic leading cone, confirmed by gated measurement (slope
ratios $1$ within errors, the factorized-transport benchmark excluded
by $\ge10$ error bars, helicity residues with chirality $-1$); at
finite momentum the node carries a parity-odd anisotropy
$c(q)\,k_xk_yk_z/|k|$, carried by every member of the four-component
coin class (exhaustive enumeration) and cancelled exactly on the
inversion-doubled Dirac
branch, whose gap is exactly $2\arctan[q_m/(1-q_m)]$ while the
off-resonant spectators of the computed 44-point Floquet--Weyl census
stay exactly massless; a media-mediated interaction obeys
$U_g=(1-2gq^2)\,U$ exactly at every cycle. The obstructions are
charted alongside: ballistic transport admits only finite velocity
sets (proven); Abelian sign gauges leave the factorized surface
(measured over a finite gauge menu); no protected two-component node
survives a table-exhaustive search. Finally, a dichotomy for the
quaternionic event class: an isotropic cone forces
overdamping---$Q=\omega_0/\Gamma\le\pi/(2\ln2)$ at every density,
$0.66$ at the working density (proven)---while in-state unitarity
forces factorized transport, under stated hypotheses on the read
geometry and the medium. The postselected two-boundary amplitudes
evade the dichotomy, exactly unitary at all times under the fresh
tape, with the positivity of their induced weight open: the exactly
relativistic object of this construction is an amplitude whose
probabilistic reading is not settled here.
\end{abstract}

\maketitle

\section{Introduction}
\label{sec:intro}

A probabilistic cellular automaton (PCA) is a classical statistical system:
a lattice of bits, a probability distribution over bit configurations, and
one local, invertible, homogeneous update rule applied in discrete time.
Wetterich~\cite{Wetterich2022,Wetterich2022b} established that a class of
such automata is \emph{exactly} equivalent to fermionic quantum field
theories in $1+1$ dimensions: the same wave function, density matrix,
operators and expectation values, at every lattice size, with no
approximation. The equivalence rests on three pillars: occupation bits map
to fermions through a Grassmann functional integral; the step operator is a
signed permutation, hence orthogonal; and a compatible complex structure
turns orthogonal into unitary evolution. The construction in three or four
spatial dimensions was left open---in Wetterich's own assessment,
``[g]eneralizations of our formalism to three or four dimensions do
not seem to encounter problems of principle. What is not yet achieved,
however, are cellular automata that realize simple fermionic quantum
field theories with Lorentz symmetry in three or four
dimensions''~\cite{Wetterich2022}; the known four-dimensional
route~\cite{Wetterich2022c} transports Lorentz-invariant four-fermion
composites rather than single fermions.

The program sits within a wider landscape. In the \emph{unitary}
quantum-cellular-automaton and quantum-walk tradition, discrete dynamics
with postulated complex amplitudes yields Weyl and Dirac kinematics in one
to three dimensions~\cite{BialynickiBirula,Meyer,DArianoPerinotti,
BisioDAriano,Arrighi,Farrelly,Kitagawa}, and Mlodinow and Brun have proven
a no-go theorem for quantum-field limits of such automata above one
dimension together with evasions that antisymmetrize distinguishable
particles in two and three dimensions~\cite{MlodinowBrunNoGo,BrunMlodinow2D,
MlodinowBrun3D}; fermion doubling in these discrete settings remains an
active subject~\cite{DoublingQW}. The present work is in the \emph{probabilistic} class, where amplitudes
are not postulated but must emerge from classical statistics---the program of
Refs.~\cite{Wetterich2021NPB,Wetterich2022,Wetterich2022b,WetterichPotential,
KreuzkampWetterich,Wetterich2025}, with conceptual roots in 't~Hooft's
cellular-automaton interpretation~\cite{tHooft}. The obstructions proven
below (finite rays, Q-pinning) are properties of this probabilistic
class; in the unitary setting, where complex amplitudes are postulated,
the corresponding constraints do not arise in this form. The dressing
and coin mechanisms below are the response to them.

This paper solves the construction half of that problem---the free
single-fermion sector with exact relativistic kinematics at leading
order---under the strictest contract: the propagating carrier is a
single-fermion occupation number; the rule set (locality, invertibility
with the unique-jump property, homogeneity, fermion-parity conservation
per block, particle--hole symmetry, no hidden continuum input) is fixed
at the outset and every construction below is verified against it
mechanically. The Lorentz-symmetric half of Wetterich's sentence is
\emph{not} delivered beyond the leading cone: Sec.~\ref{sec:chiralaniso}
proves that within this construction class the obstruction is
structural---a parity-odd velocity anisotropy, of first order in the
lattice spacing, tied to the chirality of the node---and quantifies the
partial remedies. A single deliberate symmetry relaxation underlies
this: a lone Weyl node is chiral, so the composite rule breaks spatial
inversion, axis transpositions and time reversal by design, retaining
the proper tetrahedral rotations and the antiunitary conjugation
$U(-k)=U(k)^*$. The construction is exact in the
following sense: the model's streaming schedule is chosen so that
environment world-lines outrun the carrier's light cone, no bit is ever
consulted twice, and the quenched single-particle ensemble propagator
equals the Bloch operator theory identically at every time
(Theorem~\ref{thm:freshtape}); the closed forms of
Secs.~\ref{sec:spectrum}--\ref{sec:interaction} are therefore statements
about the automaton itself.

The results divide into obstructions and constructions, and we state the
evidence class of every claim: \evP{} for symbolic or algebraic proofs,
\evM{} for exact finite verifications (tolerances quoted), \evX{} for
measurements, which are always accompanied by a known-answer gate run
through the identical analysis pipeline. The complete verification code is
public~\cite{repo}.

Obstructions (Sec.~\ref{sec:obstructions}): ballistic one-particle
transport in this class has a \emph{finite} set of group velocities, so no
three-dimensional cone is possible (Theorem~\ref{thm:finiteray}); on the
symmetric vacuum all transport speeds are rational, while a tuned vacuum
density $q$ provides a continuous, exactly stationary knob with the closed
dressing law $v=2|1-2q|$; Abelian position-periodic sign gauges relocate
band features but cannot bend the factorized (``diamond'') group-velocity
surface into a cone; and the only two-component isotropic nodes found
are isolated fine-tuned points, every one frequency locked at
$\omega_0=\pi$ and without any dial (Theorem~\ref{thm:ntwo} and the
accompanying exhaustive search).

Constructions (Secs.~\ref{sec:coincube}--\ref{sec:interaction}): a
four-component carried coin whose conversion blocks form the real
quaternion representation of $\mathrm{Cl}(0,3)$ produces an exactly
isotropic Weyl cone; the automaton realizing it (the \emph{coincube}) is a
composition of manifestly legal layers whose fermionic lift is
verified including all sign coherences; its fresh-tape streaming makes
the quenched propagator exactly the operator theory
(Theorem~\ref{thm:freshtape}), with generic slower schedules and their
re-read corrections analyzed alongside; the step operator is exactly
unitary with respect to an explicit complex structure; the local Grassmann
action is extracted in exact rational arithmetic; a dense-sweep census
with computed topological charges charts the full gapless spectrum; the
node's finite-momentum anisotropy is derived, proven unavoidable at
four components by exhaustive enumeration, and cancelled exactly on the
doubled branch (Sec.~\ref{sec:chiralaniso});
spatial-inversion doubling yields a Dirac fermion with an exact,
continuously tunable mass for the resonant node quartet; and a
media-mediated interaction with a continuous coupling is certified, its
damping law exact at every cycle.

Finally, Sec.~\ref{sec:damping} proves the second
principal result: a dichotomy. Within the quaternionic event class, an
isotropic cone \emph{forces} a nonzero, exactly mode-blind visibility
decay (the annealed step operator factorizes as a scalar times a
unitary, Proposition~\ref{prop:factor}), bounded by
$\kappa=\pi/(2\ln2)$ on the node's phase advance per unit decay and
leaving the in-state excitation overdamped
($Q=\omega_0/\Gamma\approx0.66$ at the working density); while in-state
unitarity forces factorized transport and destroys the cone. Correlated
media are included; outside the event class the bound provably fails at
an exhibited boundary. The two-boundary (in--out)
propagators---postselected amplitudes, legitimate objects in the same
transfer-matrix formalism---evade the dichotomy, under the fresh tape
at all times: they form a one-parameter unitary family with the cone
intact,
and the square-root boundary is proven to be its unique medium-blind
member.

\subsection{Physical picture: The machine in plain terms}
\label{sec:plain_terms}

Before detailing the formal algebraic and Grassmann machinery, the physical
architecture of the automaton can be stated in elementary operational terms.

\paragraph{Lattice and state space.}
The system lives on a standard three-dimensional cubic lattice $\mathbb{Z}^3$.
At each lattice site, the dynamical variables are purely classical bits divided
into two categories:
\begin{enumerate}
\item \emph{Carrier bits}: Four occupation bits per site, corresponding to a
two-bit ``coin'' register indexing four internal channels. In the
single-particle sector, exactly one carrier bit is set across the entire
lattice. These four internal states furnish the real quaternion representation
of $\mathrm{Cl}(0,3)$ that yields a single Weyl fermion. (For a Dirac fermion,
the coin register is doubled to eight channels).
\item \emph{Environment bits}: Three autonomous boolean fields per site, one
for each spatial axis ($X, Y, Z$). There are no global bits and no infinite
auxiliary tapes.
\end{enumerate}
The initial state of the environment is drawn from a Bernoulli product measure
with density $q \in (0, 1/2)$. This random initialization is the \emph{sole
source of stochasticity} in the model. For all times $t > 0$, the evolution is
strictly deterministic, reversible, and local.

\paragraph{The update cycle.}
A single time step $t \to t+1$ consists of a sequence of local, bijective
layers cycled through the three spatial axes $a \in \{x, y, z\}$. For each axis,
three consecutive operations occur:
\begin{enumerate}
\item \emph{Convert}: At every cell, if the local environment bit for axis $a$
is $1$, a fixed signed permutation matrix $C_a$ is applied to the
four-component coin register. If the bit is $0$, the coin register is unaltered.
\item \emph{Move}: The carrier hops to the nearest neighbor ($\pm 1$ cell)
along axis $a$, with the sign of the displacement governed by its coin state.
\item \emph{Stream}: The environment bits themselves scroll via local pair-swaps
along an orthogonal axis (\emph{cross-streaming}) at a velocity exceeding the
carrier's maximum speed ($v_{\text{env}} = 6$ vs.\ $v_{\text{carrier}} \le 2$
lattice units per cycle).
\end{enumerate}

\paragraph{The two essential mechanisms.}
Two structural features bridge this classical bit-machine to relativistic
quantum field theory:
\begin{itemize}
\item \emph{The fresh-tape property}: Because environment streaming outruns the
carrier's causal cone, the carrier never encounters a previously consulted
environment bit. Consequently, every conditional conversion acts on an
independent Bernoulli bit. As an exact theorem, the quenched single-particle
ensemble propagator equals the annealed Bloch evolution identically at all
times:
\begin{equation}
T_a = (1-q)\openone + q\,C_a.
\end{equation}
The ensemble average does not approximate the automaton; it describes its exact
single-particle propagator.
\item \emph{The quaternion coin}: Ballistic hopping on a cubic grid
fundamentally produces an anisotropic diamond-shaped velocity surface
(Theorem~\ref{thm:finiteray}). Spatial movement alone cannot restore rotational
isotropy on a lattice. The coincube overcomes this by internal rotation: the
conversion operators $C_x, C_y, C_z$ are chosen to satisfy the quaternion
algebra $\{C_a, C_b\} = -2\delta_{ab}\openone$. The resulting anticommutative
cross-terms cancel directional lattice artifacts near the band degeneracy,
yielding an isotropic linear dispersion $\omega = v|\bm{k}|$ with
$v = 2|1-2q|$.
\end{itemize}

Throughout, the lattice spacing and the duration of one full update cycle
are unity; frequencies are per cycle and momenta are in lattice units.

\section{Setup}
\label{sec:setup}

\subsection{States, step operator, lifts}

Configurations $\tau$ are assignments of occupation bits to species at the
sites of a cubic lattice. The probabilistic information is a real wave
function $q_\tau$ with $p_\tau=q_\tau^2$~\cite{Wetterich2022}. One time
step applies a deterministic, invertible map of configurations together
with a sign gauge: the step operator $\hat S$ is a \emph{signed
permutation} of the configuration basis, hence real orthogonal. A rule is
\emph{legal} if it is a composition of layers, each of which is (i) local,
(ii) a bijection on configurations with a unique image (the unique-jump
property), (iii) homogeneous, and (iv) fermion-parity even on its support,
so that a fermionic lift exists. The lift assigns the signs: occupation
bits map to fermionic modes in a fixed Jordan--Wigner ordering, and each
layer's signed permutation must equal the matrix of an explicit fermionic
operator in the occupation basis. All lifts below are verified at this
level, including two-particle sign coherence (Sec.~\ref{sec:coincube}).

\subsection{Environment fields and dressed carriers}

Besides the carrier species, the automata below contain autonomous binary
\emph{environment} fields: one per spatial axis (density $q$), plus one for
the mass sector ($q_m$) and one for the interaction ($g$). Environment
dynamics is free streaming by pair-swap layers, applied in fast
phase-continuing batches---the \emph{fresh-tape} schedule defined in
Sec.~\ref{sec:coincube} and used in Sec.~\ref{sec:freshtape}; a
Bernoulli product measure is exactly stationary under any pair-swap
schedule. Carriers never modify the environment except through
the explicitly constructed interaction layer of Sec.~\ref{sec:interaction}.
The vacuum is the product of the carrier vacuum with the Bernoulli
environment measure: translation invariant, clustering, with a continuous
parameter $q$ that is the central dial of the construction.

\section{Obstructions}
\label{sec:obstructions}

\begin{theorem}[Finite rays]\label{thm:finiteray}
For any legal rule with finitely many species per site that conserves
particle number in the one-particle sector---the sector with exactly one
occupied bit across all species on an otherwise empty lattice---the
one-particle Bloch operator is a monomial matrix for every momentum $k$. Consequently every dispersion branch is exactly linear in
$k$, the set of group velocities is finite, and no three-dimensional cone
$\omega=v|k|$ is possible ballistically. \evP{}
\end{theorem}

The proof is elementary---a signed permutation with momentum phases has
monomial matrix elements, and eigenvalues of monomial matrices are roots of
single monomials---but the consequence is structural: the error is
scale invariant, so no continuum limit removes it, and periodic cycles of
rotated rules do not evade it because products of unique-jump operators are
unique-jump. In one spatial dimension the light cone genuinely consists of
two rays, which is why the $1+1$-dimensional construction exists. An
adversarial test suite (seven evasion strategies with re-seeded variants,
plus reconstruction tests that fit the Bloch operator from real-space
propagation rather than constructing it) accompanies the theorem in the
repository. The same suite verifies how the construction of this paper
relates to the theorem: at fixed media the coincube cycle is a signed
permutation of the (site, channel) basis but not homogeneous, so the
theorem does not apply per realization---and indeed every realization
retains a finite velocity set. The physical configuration carries
$O(qL^3)$ occupied environment bits, so the propagating object is a
carrier dressed by an extensive environment, not the theorem's
one-particle state, and the cone exists only in the media-averaged
propagator, which is not monomial. The construction evades the
theorem's hypotheses, not its conclusion.

Two further results constrain the construction. First, on the symmetric
(density-$\tfrac12$) vacuum the transport speeds of the
autonomous-streaming subclass of conditional rules~\cite{Wetterich2022b}
are rational \evP{} (frozen-vacuum and fully coupled rules are outside
this statement); the vacuum density $q$ is
the legal continuous parameter, with the closed dressing law, for the
representative rule used throughout,
\begin{equation}
v(q) = 2\,|1-2q| ,
\label{eq:vq}
\end{equation}
proven as a two-class Markov-additive limit and verified against exact ring
enumeration as polynomial identities in $q$ \evP{}. Second, dressing alone
yields a factorized walk whose group-velocity surface is the $\ell^1$ ball
(the \emph{diamond}): slope ratios along $(110)$ and $(111)$ relative to
$(100)$ equal $\sqrt2$ and $\sqrt3$. Abelian position-periodic sign gauges
(Kogut--Susskind axis phases~\cite{KogutSusskind,Susskind}, per-stall signs, and
their products) relocate the nodal points and modify residues but leave the
ratios on the diamond, measured as
$r_{110}/r_{111}=1.40/1.68$ (base and stall gauges, identical spectra) and
$1.46/1.82$ (KS gauges) \evX{}---an early ungated instrument whose absolute
anchor deviates $3.5$--$5.5\%$ from the exact operator, quoted only for the
on-diamond pattern; a carried non-Abelian structure is therefore necessary.

\begin{theorem}[Two-component exclusion]\label{thm:ntwo}
Consider legal automata whose carrier has a two-dimensional internal
space: per-axis cycles of the monomial factors $(1-p)E_a+pC_a$ or
$E_a[(1-p)\openone+pC_a]$, where $E_a=\mathrm{diag}(e^{\pm ik_a})$ and
$C_a$ is a real signed permutation, with one or two (blocked)
engagements per axis.
At every conversion weight $p\neq\tfrac12$ the single-engagement
cycles have no degenerate propagating pair at any momentum (rotation
coins force $(1-p)^2=p^2$; reflection coins never produce one; diagonal
coins degenerate). \evP{}
\end{theorem}

The remaining cases are closed by search---exhaustive in the discrete
tables, heuristic in momentum, sampled in weight---not by proof.
The only isotropic two-component Weyl nodes found in this class are
\emph{isolated fine-tuned points}: the single-origin placement at
exactly $p=\tfrac12$
($\omega_0=\pi$, $v=1$, $\chi=-1$, $|\lambda_0|=2^{-3/2}$, all forced,
width $\Gamma=\tfrac32\ln2$ at the chord midpoint---the maximally
damped point of Theorem~\ref{thm:qbound}) \evP{}; and the blocked
two-origin cycles at numerically located tuned weights ($p^*=0.3300$
additive, $p^*=0.3204$ multiplicative; anisotropy zero within
extrapolated residuals $\le8.5\times10^{-4}$), both likewise locked at
$\omega_0=\pi$ and heavily damped, with computed chiralities $\chi=-1$
and $\chi=+1$ respectively \evM{}. At the sampled weights
$p\in\{0.15,0.25,0.35,0.5\}$ every design has a strictly positive
anisotropy floor \evM{}; the floors necessarily dip toward zero on
approach to the tuned $p^*$, which lie between the sampled weights.
(In the \emph{unitary} walk setting, where amplitudes are postulated,
two-component three-dimensional constructions
exist~\cite{Chandrashekar,DArianoErba}; the exclusion here is specific
to the probabilistic class.)

Two-component constructions therefore admit no \emph{protected}
relativistic carrier: every isotropic node found is an isolated point in
parameter space, its quasienergy locked at $\omega_0=\pi$---no mass dial,
no width dial, no vacuum-density dressing. The minimal internal space
with a symmetry-protected node on an entire parameter family is
four-dimensional and real, where the Clifford algebra $\mathrm{Cl}(0,3)$
has its quaternionic representation; that is the construction of the next
section.
\section{The coincube automaton}
\label{sec:coincube}

\subsection{Layers}

The carrier has four species per site: a two-bit coin $(b_1,b_2)$, channel
index $c=2b_1+b_2$. Writing $X$, $Z$ for the real Pauli matrices and
$XZ$ for their (antisymmetric) product, define the conversion triple and
direction assignments
\begin{align}
C_x &= XZ\otimes\openone, & d_x &= \mathrm{diag}(Z\otimes\openone),
\nonumber\\
C_y &= Z\otimes XZ, & d_y &= \mathrm{diag}(Z\otimes Z),
\nonumber\\
C_z &= -X\otimes XZ, & d_z &= \mathrm{diag}(\openone\otimes Z).
\label{eq:tables}
\end{align}
The $C_a$ are pairwise anticommuting with $C_a^2=-\openone$: the real
quaternion representation of $\mathrm{Cl}(0,3)$ \evM{}. Each full cycle
applies, for each axis $a$ and each of two block origins, three layers:

\emph{L2 (conversion).} At every site, controlled on the axis-$a$
environment bit at that site, the two channel pairs of $C_a$ are swapped
(a one-site block; two bit flips, parity even).

\emph{L1 (motion).} Channel $c$ translates by $d_a(c)\in\{\pm1\}$ along
axis $a$, unconditionally.

\emph{L3 (environment).} The axis-$a$ environment field streams along
$s_a=a{+}1\ (\mathrm{mod}\ 3)$: each axis-$a$ substep applies
\emph{three} consecutive pair-swap layers to the field, with origins
that continue the field's global alternation phase (the $n$-th swap ever
applied has origin $n\bmod2$, counted across substeps and cycles).
Each layer is the same certified environment layer; the schedule is
deterministic, homogeneous and autonomous. Phase continuation is
required: a schedule restarting each batch at origin $0$
repeats an origin at the batch boundary and the repeated pair cancels
(measured transport $0$ per batch, versus exactly $3$ for the
phase-continuing batch) \evM{}. Under this schedule every environment
bit moves at exactly $\pm6$ sites per cycle along $s_a$, in a direction
fixed at $t=0$ by its parity class, while the carrier moves at most $2$
sites per cycle per axis---the environment outruns the carrier, with the
consequences of Sec.~\ref{sec:freshtape}. Two design constraints are
mandatory. \emph{Cross-streaming} ($s_a\neq a$): if the field streams
along its own axis, a co-moving carrier re-reads its own bit,
conversions arrive in nearly scalar bursts ($C_a^2=-\openone$), and the
measured Weyl pair splits into one damped-real and one oscillating
pole---an exceptional-point transition of the damped band structure in
the sense of non-Hermitian topology~\cite{Bergholtz}---and the node is
destroyed \evX{}. The second is even linear size $L$ (enforced in
code); all statements below assume even $L$.

Each layer is manifestly legal; the composition defines the automaton.

\subsection{Fermionic lift and sign coherence}

The lift of L2 is an environment-controlled pair of $\pm90^\circ$ Givens
rotations $G=\exp[\tfrac{\pi}{2}(a^\dagger_c a_{c'}-a^\dagger_{c'}a_c)]$
acting on same-site modes: a signed permutation of the Fock basis,
parity-even, identity at empty control, with single-particle matrix equal
to the corresponding signed-swap block of $C_a$ and $G^2=-\openone$ on the
one-particle sector; the quaternion sign structure \emph{is} the fermionic
rotation sign \evM{} (dense verification, all axes). Translations and
environment swaps lift with the standard reordering parity. The composite
cycle is sign coherent: in the segregated mode ordering (all carrier modes
before all environment modes) the full lifted cycle is a signed permutation
whose one-particle sector equals the stated rule and whose two-particle
sector equals the exact antisymmetrized square $\Lambda^2 M_1$---verified
in genuine three-dimensional geometry on all 5778 pair states at linear
size $L=3$, with mutation controls (deliberately corrupted sign rules fail
the check) \evM{}. At $L=2$ the two
sub-steps per axis make net translation crossings even, so $L=2$ is blind
to translation-sign errors; three-dimensional sign claims must be checked
at $L\ge3$.

\section{The fresh-tape property}
\label{sec:freshtape}

The environment schedule of L3 was chosen so that no bit is ever
consulted twice:

\begin{theorem}[Fresh tape]\label{thm:freshtape}
Consider the coincube in the single-particle sector with the L3 schedule
of Sec.~\ref{sec:coincube} and the Bernoulli($q$) product vacuum.
(a)~On the infinite lattice, no conversion history reads the same
physical environment bit twice, at any time. Consequently the quenched
ensemble propagator equals the annealed propagator \emph{identically},
$G_{\rm quenched}(t)=G_{\rm annealed}(t)$ for all $t$: every closed form
derived from the annealed Bloch operator $U(k)$---the node census, the
propagator factorization, the linear conversion law, the mass law---is
exact for the automaton itself. (b)~On the $L^3$ torus ($L$ even) the
same holds for every evolution of $T\le\lceil L/8\rceil$ cycles, and
this horizon is sharp: the first wrap-induced re-read joins two reads
separated by exactly $2\lceil L/8\rceil$ of the field's own axis-$a$
substeps. \evP{}, \evM{}
\end{theorem}

The proof (Appendix~\ref{app:proofs}) is a kinematic argument on three
finite lemmas; every lemma and the sharpness
of (b) are machine-checked, and the equality in (a) is verified as an
exact path sum over \emph{all} conversion histories, with physical bit
labels tracked through the real streaming permutations: all four launch
channels at $T=3$ ($2^{18}$ histories each, deviation
$2.8\times10^{-17}$) and channel $0$ at $T=4$ (all $2^{24}$ histories,
deviation $9.1\times10^{-17}$) \evM{}. The scope is every observable
linear in the evolved field---amplitudes, propagators, and the spectra
derived from them; quadratic estimator statistics (variances of
paired-media estimators) still distinguish quenched from annealed
ensembles, and nothing in the theorem concerns them.

\subsection{Generic schedules and re-read corrections}
\label{sec:generic}

The fresh-tape property belongs to the schedule, not to the layer type.
The generic single-swap schedule (one pair-swap layer per axis substep,
origins alternating per substep) transports bits at $\pm2$ sites per
cycle on the two parity sublattices---pair-swap streaming
counter-propagates; it is not a unit shift---which exactly matches the
carrier's maximal transverse speed: carrier and bit can co-move, and the
same bit is consulted again after one or two cycles. In the exact
$T=3$ path sum this produces $13\,654$ re-read branch events and a
quenched-vs-annealed propagator deviation of $1.7\times10^{-3}$; a
separate twenty-cycle kinematic certificate counts $380$ re-read pair
classes for this schedule against zero for the fresh-tape schedule
\evM{}. At production scale
the re-read corrections are directly measurable: under the generic
schedule the quenched node parameters shift from the operator values by
$+2.8\%$ ($|\lambda_0|$) and $+8.6\%$ ($\omega_0$) at $q=0.08$ (of
which only $+0.3\%$/$+2.3\%$ is instrumental, from the gate row), and
the interaction damping law acquires excursions beyond the first cycle
($2$--$2.5\%$ at $g=0.3$, up to $6\%$ at $g=0.6$; measured at $L=12$
inside the generic run's own wrap regime, so re-read and torus-wrap
contributions are not separated there) \evX{}. Under the fresh-tape schedule every one of these
corrections vanishes identically (Theorem~\ref{thm:freshtape}); the
measured comparison is Fig.~\ref{fig:schedules}.

\begin{figure*}[!t]
\includegraphics[width=0.85\textwidth]{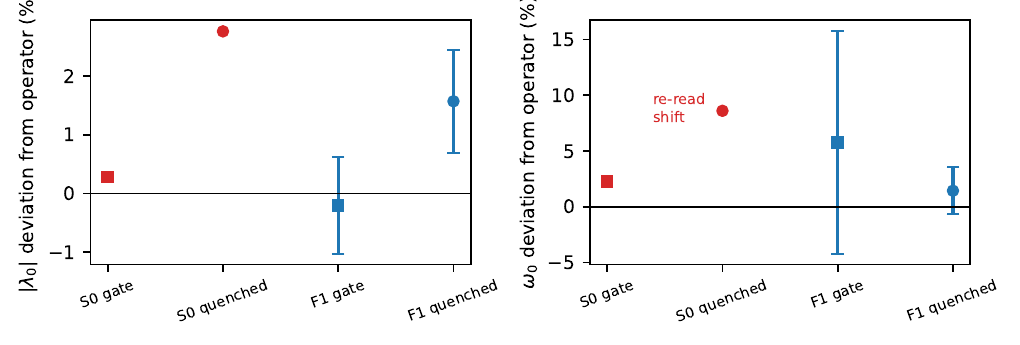}
\caption{The fresh-tape identity at scale: measured quenched node
modulus (left) and frequency (right) as deviations from the exact
operator values, under the generic schedule (red; re-read shifts, gate
rows shown for the instrumental part) and under the fresh-tape schedule
(blue, with block-jackknife errors). The generic-schedule points predate
per-node jackknife storage and are shown as central values; their
instrument's offsets are the gate points. \evX{}}
\label{fig:schedules}
\end{figure*}

\section{Free spectrum: the Weyl cone}
\label{sec:spectrum}

\subsection{Closed forms}

By Theorem~\ref{thm:freshtape} every engagement reads a fresh
Bernoulli($q$) bit, so the one-particle transfer per sub-step is exactly
$T_a(k)=E_a(k_a)\,[(1-q)\openone+qC_a]$ with
$E_a=\mathrm{diag}(e^{ik_ad_a})$, and the cycle operator is
$U(k)=T_z(k_z)^2T_y(k_y)^2T_x(k_x)^2$ (axes applied in the order
$x,y,z$; the closed forms below hold for either orientation of the
cycle, the census numbers are specific to this one): the operator
theory of this section is the
automaton's single-particle dynamics, not an idealization of it. (The
same operator arises for any schedule under independent
per-engagement reads---the annealed average---which is how the generic
schedules of Sec.~\ref{sec:generic} are compared against it.)

\begin{proposition}[Scalar--unitary factorization]\label{prop:factor}
Each conversion block $C_a$ is real, antisymmetric and orthogonal, so
$T_a(k)^\dagger T_a(k)=[(1-q)^2+q^2]\,\openone$ identically in $k$;
hence
\begin{equation}
U(k)=\rho^6\,V(k),\qquad \rho^2=(1-q)^2+q^2,
\label{eq:factor}
\end{equation}
with $V(k)$ exactly unitary for every $k$. \evP{}
\end{proposition}

Proof:
$(\alpha\openone+\beta C)^\dagger(\alpha\openone+\beta C)
=(\alpha^2+\beta^2)\openone+\alpha\beta(C+C^T)
=(\alpha^2+\beta^2)\openone$; as a machine check, every eigenvalue
modulus equals $\rho^6$ to $10^{-15}$ over random momenta at four values
of $q$ \evM{}. Two closed forms are immediate: the node modulus
$|\lambda_0|^2=\rho^{12}=[(1-q)^2+q^2]^6$, and the exact periodicity
$U(k+\pi e_a)=U(k)$, since $E_a(k_a+\pi)=-E_a(k_a)$ and the sign cancels
in $T_a^2$ \evP{}. The linear splitting of the node doublet is proven
symbolically to be exactly isotropic, identically in $q$, with speed
$v(q)$ rising from $2/\sqrt3$ at $q\to0$ to a maximum $1.207$ near
$q=0.30$ ($v(0.08)=1.1806$) \evP{}; numerically the splitting isotropy
holds to $2\times10^{-4}$ over 50 random directions after $h\to0$
extrapolation \evM{}.

\subsection{The chiral anisotropy}
\label{sec:chiralaniso}

The exact isotropy of the node is a property of the leading order
alone. The splitting is
\begin{equation}
s(k)\;=\;2v\,|k|\;+\;c(q)\,\frac{k_xk_yk_z}{|k|}\;+\;O(|k|^3),
\label{eq:oddterm}
\end{equation}
the second term proportional to the lowest anisotropic invariant of
the chiral tetrahedral group: to this order the slope depends on the
direction only through $\hat k_x\hat k_y\hat k_z$ (exactly, for
$A_4$-related directions; to $O(|k|^4)$ otherwise), and
$c=42.7,\,20.0,\,10.5,\,8.6$ at $q=0.08,\,0.15,\,0.25,\,0.30$, growing
toward small density \evM{}. The quaternionic triple carries an
orientation ($C_xC_yC_z=+\openone$, so the composed cycle order gives
$C_zC_yC_x=-\openone$), so no signed coin relabeling
implements an odd axis permutation; the spectral symmetry group is
exactly $T\cong A_4$ (Sec.~\ref{sec:census}), and $A_4$ admits the
degree-3 invariant $k_xk_yk_z$: the algebra that forces the isotropic
leading cone forbids every symmetry that would exclude the odd term.
The term is a property of the entire class: the $4\times4$
signed-permutation group contains exactly $12$ elements squaring to
$-\openone$ (the left- and right-multiplication quaternion families)
and $48$ ordered anticommuting triples, and over every triple,
signed axis assignment, block ordering and balanced direction-table
combination that preserves an isotropic-leading node, the odd
coefficient never falls below $0.85$ of the production value \evM{}.
Legal deterministic counterterm layers---coin-conditional
transverse translation dipoles, which read no environment bits and
leave Theorem~\ref{thm:freshtape} untouched---couple to the
correct tensor structure with bounded strength: the best family
member removes $45\%$ of the coefficient, with the leading isotropy
intact and the residual still of rank 3 \evM{}.

Three consequences are quantitative. \emph{Scale}: the relative
velocity anisotropy over the sphere is
$c\,|k|/(3\sqrt3\,v)\simeq7.0\,|k|$ at $q=0.08$, of order one at the
smallest momentum of the $L=48$ production torus, and below $1\%$ only
for $L\gtrsim4400$: the isotropic cone is a $k\to0$ statement, and the
approach to it is one power of the lattice spacing slower than for a
factorized lattice fermion. \emph{The instrument}: the cone estimator
of Sec.~\ref{sec:quenchedcone} pools every direction with its full
$\pm$ star, and the star average annihilates each odd term identically
\evP{} (verified to $2\times10^{-17}$ \evM{}), so the measured ratios
are statements about the even sector and the $k\to0$ slope; the
direction-resolved ratios of the exact operator at the probe radii are
$r_{111}=1.116,\,1.207,\,1.366$ at $\delta=0.03,\,0.05,\,0.08$.
\emph{Cancellation}: the term is odd under spatial inversion, and the
inversion-doubled model of Sec.~\ref{sec:mass} cancels it exactly, at
every momentum and mass density. The term is tied to the net
chirality: an eight-component chirality-2
construction with full proper-octahedral covariance (two
same-orientation blocks exchanged by a legal block swap) removes it
from the dispersion only by moving it into the splitting of a nested
double cone that no coupling can close \evM{}; it vanishes at first
order precisely when the net chirality does.

\subsection{The gapless census}
\label{sec:census}

The gapless set of the unitary factor $V(k)=U(k)/\rho^6$ is charted by
a dense sweep of the reduced zone $[0,\pi)^3$: minimum pairwise
eigenvalue distances on a $64^3$ grid, Nelder--Mead refinement of every
local minimum below $0.15$ to machine precision, sphere-probe
classification of each refined degeneracy (an isolated point keeps a
finite gap in all directions; a line member has an antipodal zero
pair), and orbit grouping under the spectral symmetry group. That group
is determined empirically: of the 48 signed coordinate permutations,
exactly the 12 proper rotations of the chiral tetrahedral group
$T\cong A_4$ preserve the spectrum, and the 12 complementary operations
(including $k\to-k$) map it to its conjugate---no odd permutation
appears, reflecting the cyclic $x\to y\to z$ layer ordering. For every
isolated node the chirality is \emph{computed} as a Chern number: the
Fukui--Hatsugai lattice Berry flux of the upper doublet member over an
enclosing sphere, built from the spectral projectors of $V$ (well
defined by Proposition~\ref{prop:factor}) and calibrated on
$\exp(+ik\cdot\sigma)$; every quoted charge is an exact integer, stable
under two sphere radii and two angular grids \evM{}.

\begin{table}[!htb]
\caption{The gapless census of the coincube walk (chord weights,
$q=0.08$): five point orbits (44 isolated Weyl points), three nodal
lines, and their triple junction. Quasienergies are $+$Im-branch values;
charges are computed Chern numbers per node.}
\label{tab:census}
\begin{ruledtabular}
\footnotesize
\begin{tabular}{lcccc}
orbit & rep.\ $k/\pi$ & $\omega/\pi$ & mult. & $\chi_{+}/\chi_{-}$\\
\colrule
corner & $(0,0,0)$ & $0.1006$ & $1$ & $-1$ / $+1$\\
half-points & $(\tfrac12,0,0)$+perms & $0.9221$ & $3$ & $+1$ / $-1$\\
$\langle111\rangle$ octet A & $(0.265,0.265,0.735)$ & $0.5142$ & $8$ & $-1$ / $+1$\\
$\langle111\rangle$ octet B & $(0.243,0.243,0.757)$ & $0.5652$ & $8$ & $+1$ / $-1$\\
generic orbit & $(0.028,0.401,0.902)$ & $0.9479$ & $24$ & $-1$ / $+1$\\
nodal lines & $(t,\tfrac12,\tfrac12)$+perms & $t$-dep. & $3$ & ---\\
$R$ junction & $(\tfrac12,\tfrac12,\tfrac12)$ & $1$ & $1$ & $U=-\rho^6\openone$
\end{tabular}
\end{ruledtabular}
\end{table}

Table~\ref{tab:census} is the census. Beyond the corner and the
half-points the zone carries 40 further Weyl points: two
$\langle111\rangle$ octets with members $(\pm t,\pm t,\pm t)$ mod $\pi$
and a fully generic 24-point orbit. Multiplicities count distinct
momenta; the corner and half-points host both branches at each
momentum, while the octets and the generic orbit host one
band-touching event per point, so the per-branch event counts are
$(1,3,4,4,12)$ and the point charges on the $+$Im branch sum to
$S_+=-1+3-4+4-12=-10$ (and $S_-=+10$)---a bookkeeping total over the
$+$Im quasienergy window, which is a convention, not a band; the
invariant statement is the conjugation pairing below. Every other
local minimum of the grid-basin search refines to a nonzero gap
(smallest avoided crossing $0.019$ at $q=0.08$, $0.050$ at $q=0.15$,
and $0.0013$ at $q=0.30$---the weakest completeness margin, at the
density where the topology itself changes): to those resolutions the
lists are complete. The same five-orbit topology with
identical charges holds at chord $q=0.15$ and at the two-boundary
(arc) weights at $q=0.08$---the extra species are not artifacts of the
damping factorization---but it is not permanent in $q$: the generic
orbit annihilates pairwise on the $\{k_a=0\}$ planes at
$q_c\in(0.292,0.294)$, leaving a 20-point census with $S_+=+2$ at
$q=0.30$ \evM{}.

Two facts govern the charge balance. First, the eigenphases of $V(k)$
cover the entire circle---a 360-bin occupancy histogram over a $40^3$
momentum grid has no empty bin---so $V$ has no Floquet quasienergy gap
anywhere: the static Nielsen--Ninomiya~\cite{NielsenNinomiya} argument
applies to no branch or quasienergy window, while the Floquet-extended
zero-sum rule~\cite{HigashikawaUeda,BesshoSato}, which fixes the total
chirality of a gapless unitary at its bulk winding number $W_3[V]$, is
satisfied identically: all layers are real, so $V(-k)=V(k)^*$, under
which $W_3$ is odd, hence $W_3[V]=0$ \evP{}---precisely the
conjugation-pairing total below. $S_+=-10$ violates nothing. The one
exact global constraint is that pairing: $U(-k)=U(k)^*$ pairs every
event $(k,\omega,\chi)$
with $(-k,-\omega,-\chi)$---verified event by event---and the
two-branch total vanishes identically. Second, the nodal lines: the
loop Berry phase around a line is close to $\pi$ but \emph{not}
quantized ($0.954\pi$--$0.998\pi$ along the line; no protecting
symmetry was identified), and no closed Gauss surface enclosing one
line while avoiding the gapless set exists, since the three lines
intersect at the $R$ point; the loop-phase winding along the full line
is consistent with zero at the smaller tube radius but is
radius-unstable, so no monopole charge is assigned to the lines---and
none is needed: the only surviving global constraint, $W_3[V]=0$, is
already saturated by the point charges' conjugation pairing \evM{}.

All corner-local measurements below probe the doublet at $k=0$,
$\omega_0=0.1006\pi$. The nearest other gapless feature of any kind
lies at torus distance $0.4134\pi$ from the corner (a generic-orbit
member; the octets follow at $0.4216\pi$ and $0.4590\pi$, the
half-points at $\pi/2$, the nodal lines at $0.7071\pi$), and the
nearest spectator point node in quasienergy is $0.414\pi$ away, so
corner-local fits at $|k|\lesssim0.15$ are separated from every
spectator by more than two decades in momentum. But the walk is a
lattice fermion with a larger spectator census than a corner count
suggests: its additional species are displaced, not
absent~\cite{DoublingQW}.

\subsection{Quenched measurement}
\label{sec:quenchedcone}

On the streaming (quenched) environment the propagator is measured from the
exact per-medium single-particle signed field, averaged over media. The
instrument fits the full $4\times4$ one-cycle Bloch matrix per momentum
from all four basis launches, pools symmetry-equivalent momenta at the
level of characteristic-polynomial coefficients (basis invariant, hence
unbiased) before rooting, probes momenta off the lattice grid inside the
node's linear range, and extrapolates $\delta k\to0$. The $\pm$-star
pooling annihilates the odd anisotropy of Sec.~\ref{sec:chiralaniso}
identically, so all ratios below probe the even sector and the $k\to0$
slope. Every production run
carries an annealed known-answer row through the identical pipeline; a run
whose gate row fails is discarded. Near-degenerate pole extraction is the
dominant threat (eigenvalues of a noisy nearly defective matrix split as
$\sqrt{\text{noise}}$), which this design removes.

\begin{figure}[!htb]
\includegraphics[width=\columnwidth]{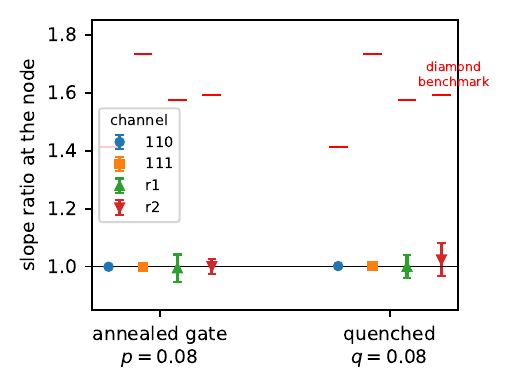}
\caption{Slope ratios at the node, annealed gate row and quenched row,
against the factorized-transport benchmark. Error bars are
leave-one-block-out jackknives over ten independent medium blocks; by
Theorem~\ref{thm:freshtape} the two rows measure the same operator.
\evX{}}
\label{fig:cone}
\end{figure}

Results (Fig.~\ref{fig:cone}): the annealed gate row reproduces the
exact operator (cubic-channel ratios within $10^{-3}$ of unity;
off-symmetry channels within their own statistics; node modulus
$0.619\pm0.005$ vs $0.620$, frequency $0.334\pm0.032$ vs $0.316$), and
the quenched row lands on the same
operator, as Theorem~\ref{thm:freshtape} demands: node modulus
$0.6300\pm0.0054$ ($+1.6\%$, $1.8\sigma$; an independent-seed
replication at doubled $R$ returns $0.6180\pm0.0038$, $-0.4\%$---the
excursion is statistical) and frequency
$0.3205\pm0.0067$ ($+1.4\%$, $0.7\sigma$) against the closed forms,
with slope ratios $r_{110}=1.002\pm0.002$, $r_{111}=1.003\pm0.003$,
$r_{\mathrm a}=1.000\pm0.040$ and $r_{\mathrm b}=1.024\pm0.057$ for the
two off-symmetry directions. The factorized-transport (diamond)
benchmark is excluded in every channel---an exclusion of the benchmark
\emph{value}, not of a competing automaton, since a fully factorized
walk supports no measurable doublet at all (Sec.~\ref{sec:damping}):
the weakest exclusion---by construction the off-symmetry channel with
the fewest orbit members---sits $10.0$ of its error bars from the
diamond value, and the cubic channels sit more than $160$ error-bar
widths away. The estimator's power to resolve anisotropy is established
by the positive controls: end to end on the off-symmetry channels (a
legal broken-coin walk under the same schedule, recovered within
errors, anisotropy resolved at $8.6\sigma$/$11.6\sigma$), and for the
cubic stars---which orbit pooling symmetrizes for cubic-breaking
truths---against a synthetic diamond-splitting truth through the
identical analysis. Absolute slopes carry the large medium-ensemble and
fit-window uncertainty seen throughout (block jackknife $\pm0.08$ on
the gate row and $\pm0.16$ on the quenched row); the
two rows scatter around the operator value $1.18$ accordingly, and the
node parameters and ratios are the precision observables. \evX{}

\subsection{Helicity residues}

At the node the effective generators $h_a$ on the doublet satisfy the exact
Clifford algebra $\{h_a,h_b\}=2v^2\delta_{ab}$; because the node frame of
the damped (non-normal) operator is not unitary, the Pauli coefficients of
$h_a$ are complex with \emph{bilinear} orthogonality
$R^TR=v^2\openone$ ($R\in O(3,\mathbb C)$), and
$\det R=-v^3$ exactly: chirality $\chi=-1$ \evM{}. The split branches'
spectral projectors equal the helicity projectors
$\tfrac12(\openone\pm n\cdot\sigma)$, $n=Ru/\sqrt{(Ru)\!\cdot\!(Ru)}$.
Measured on the quenched vacuum (the 26 directions of the cubic
$\{100\}/\{110\}/\{111\}$ stars at each of the three---by
$\pi$-periodicity identical---corner representatives; gate row passed):
chirality $-1$, stable under every jackknife resampling in both rows;
pointwise helicity-direction agreement $|1-n\cdot n_{\rm exact}|$ of
mean $0.010\pm0.002$, with the gate row returning $0.013\pm0.008$ on
the same operator. The bilinear isotropy statistic (a maximum over the
nine entries of $R^TR/s-\openone$, hence positively biased) is reported
only as an upper bound: the quenched row gives $0.060\pm0.018$ against
the gate row's $0.065\pm0.017$---the two rows are statistically
identical, as Theorem~\ref{thm:freshtape} demands, and both sit at the
instrument's noise floor, bounding the node's bilinear anisotropy at
$\lesssim0.1$ \evX{}. (The residue
projectors are rank one by construction of the spectral estimator, so
``purity'' is not reported as a measurement.) One convention is essential for eigenvector work: the
$e^{-ik\cdot x}$ Fourier contraction estimates $U(-k)=\bar U(k)$, whose
upper-half eigenvectors belong to the opposite-helicity conjugate doublet;
residue measurements must use $e^{+ik\cdot x}$.

\section{Exact quantum mechanics}
\label{sec:exactqm}

\subsection{Complex structure}

Let $P$ be the lifted carrier particle--hole conjugation, the Majorana
string $\prod_i(a_i+a_i^\dagger)$ over carrier modes. In the segregated
ordering, $P$ is real orthogonal with $P^2=+1$ (carrier mode number
$M=4L^3\equiv0 \bmod 4$), and
\begin{equation}
[\hat S, P]=0
\end{equation}
\emph{exactly, including all lift signs}, as an operator identity on the
full Fock space at every lattice size \evP{}: a string-factorization
lemma reduces each layer's commutation with the Majorana string to a
fixed finite block check (conversion and imprint blocks act on same-site
modes, so their Jordan--Wigner spectators are block internal), and the
translation lift contributes
$\mathrm{sign}(\pi)=(-1)^{(L-1)L^2}=+1$ for every $L$
(Appendix~\ref{app:proofs}). Independent sector certificates confirm the
identity directly through the three-carrier sectors at $L=2,3$
($204\,156$ three-carrier states at $L=3$, $1.6\times10^6$
state--draw evaluations; at odd $L$, where the even-$L$ pair-swap
batch is undefined, the certificate streams by a single plane
transposition---the sign identities do not depend on the
schedule), the $M{-}2$ sector by
particle--hole duality, and one- and two-particle sectors at $L=4$, over
eight environment draws \evM{}. With the grading $\eta=\mathrm{sign}(N_c-M/2)$---conserved
by every layer, odd under $P$, and size independent away from half
filling---define
\begin{equation}
K=\mathrm{diag}(\eta), \qquad I = P\,\mathrm{diag}(\eta).
\end{equation}
Then $K^2=1$, $I^2=-1$, $\{K,I\}=0$, $I$ antisymmetric, and $[\hat S,I]=0$;
with multiplication $(\alpha+i\beta)q:=\alpha q+\beta Iq$ and the induced
Hermitian form, $\hat S$ is a unitary operator on a complex Hilbert
space~\cite{Wetterich2022}: the automaton's quantum mechanics is exact
\evM{}. On the substrate the construction extends to the half-filled sector
by an explicit orbit-wise grading; in three dimensions no size-independent
grading of the half-filled sector was found and the question is left open.
Away from half filling, all carrier sectors are covered at every size, by
proof and by direct verification.

\subsection{Bridge to the spectral analysis}

The complex structure above is global (a Majorana string), and the
spectral results of Sec.~\ref{sec:spectrum} use the ordinary Fourier $i$.
These are the same quantum mechanics, by an explicit intertwiner: the map
$\Phi(u+iv)=u\oplus(-Pv)$ is a complex-linear isomorphism from the
complexified one-particle sector with the Fourier $i$ onto the
$(V_1\oplus V_{M-1},I)$ eigensector of the automaton's Hilbert space,
intertwining the dynamics (this is exactly $[\hat S,P]=0$ restricted to
the sector) and the Hermitian forms. Consequently the measured propagator
\emph{is} an $I$-Hermitian matrix element, every $\pm\omega$ conjugate
pair of the Fourier analysis appears once as a particle and once as an
antiparticle in the $(K,I)$ picture, and the chirality bookkeeping of
Sec.~\ref{sec:spectrum} is the particle/antiparticle structure of a
single Weyl node. Two scope remarks are owed. The complex structure is
nonlocal---a Majorana string mapping $k$-particle to $(M{-}k)$-particle
sectors---so no local occupation observable commutes with $I$; the
identification of measured objects with $I$-Hermitian matrix elements is
established here for the one-particle sector, and multi-carrier
correlators (such as the $C_2$ certificates of
Sec.~\ref{sec:interaction}) are not re-expressed in the $I$ picture.
The bridge's content follows from $[\hat S,P]=0$, which is proven at
every size (Appendix~\ref{app:proofs}); the direct machine check---at
$L=3$, a sign-sensitive size (Sec.~\ref{sec:coincube}), over three
environment draws including the working density---confirms the
intertwining identities exactly (deviation zero at machine level) and
the form isometry to $9\times10^{-15}$ \evP{}, \evM{}. What remains genuinely distinct is the damping: the
in-state reduced dynamics is a contraction (the chord), and only the
two-boundary amplitudes of Sec.~\ref{sec:damping} are unitary (the arc);
the bridge identifies the kinematics, not the dissipative structure.

\subsection{Grassmann action}

The local factors and actions of all blocks are extracted mechanically in
exact rational arithmetic by the pipeline of
Ref.~\cite{Wetterich2022}, Eqs.~(51)--(58), validated term by term on the
closed-form interaction action of Ref.~\cite{Wetterich2022b}%
\footnote{The validation surfaced one typographical error in
Ref.~\cite{Wetterich2022b}: the coefficient of the top monomial in the
closed-form interaction action, Eq.~(38) there (source label CS9;
equation number confirmed by compiling the arXiv source), must be $-2$,
not $-1$; the printed value fails the defining relation $e^{-L}=K$ at
that single monomial. The recomputation is committed as
\texttt{wetterich\_eq38\_check} \evM{}.}. For the coincube conversion blocks the physical (Givens-lift)
gauge yields an axis-universal action: 40 terms with identical degree
structure for all three axes; the quadratic part is the signed-swap
bilinear of $C_a$ plus the environment transport term, and the
environment control appears as quartic and higher couplings between the
channel bilinears and $e^\prime\tilde e$ \evM{}. Axis universality is a
property of the physical gauge specifically: the bare (all-plus sign)
gauge yields axis-dependent term counts ($49/42/51$)---discarding the
lift signs breaks the relabeling equivalence, a gauge artifact asserted
as such \evM{}. The streaming layers contribute the standard quadratic
transport bilinears; one fresh-tape batch's lift equals the composition
of its three single layers sign for sign, and its action is again a
single transport bilinear of the composed permutation \evM{}.

\section{Mass: inversion doubling}
\label{sec:mass}

Doubling the coin with a mass bit $b_m$ whose value \emph{reverses all
direction assignments}, $d^{(8)}_a=d_a\otimes\mathrm{diag}(1,-1)$, with
conversions blind to $b_m$, places two opposite-chirality copies of the
node at the same momentum and frequency (the $b_m{=}1$ sector is the
spatial-inversion image; sector chiralities $\mp1$ verified) \evM{}. The
alternative doubling $C_a\otimes Z$ fails: $-C_a$ generates the
opposite triple orientation, which is the \emph{other} corner frequency
family, so the two sectors share no node. A mass layer
$M=(1-q_m)\openone+q_m C_m$ with $C_m=\openone_4\otimes XZ$ (a
$b_m$-flipping controlled Givens, driven by a fourth environment field of
density $q_m$, streamed at three phase-continuing swaps per cycle along a
fixed axis so that its reads are also fresh---a once-per-cycle field
outruns the carrier iff its speed is at least $3$, and the fresh-tape
identity then covers the massive model; on the torus the mass tape
advances three sites per cycle against a carrier bound of two, so its
first wrap re-read requires a cycle separation of $\lceil L/5\rceil=10$
at $L=48$, beyond the $T=6$ window used; the odd per-cycle batch count
makes the massive and interacting rules homogeneous with period two in
time, the free rule remaining period one) couples the chiralities. Two operator identities carry the
sector \evP{}: $U(k)=(\openone\otimes R_m)[U_4(k)\oplus U_4(-k)]$ and, at
the node, $U(0)=U_4(0)\otimes M_2$ with $M_2$ a scaled rotation of the mass
bit; hence the gap is \emph{exactly}
\begin{equation}
2m \;=\; 2\arctan\!\frac{q_m}{1-q_m},
\label{eq:massgap}
\end{equation}
independent of $q$, with the multiplet center unshifted; frequencies add
angles as in a telegraph process, and $m=q_m+O(q_m^2)$ recovers the Kac
form~\cite{Kac}. On the node space the generators obey the full Dirac--Clifford
algebra $\{h_a,h_b\}=2v^2\delta_{ab}$, $\{h_a,\beta\}=0$, $\beta^2=1$,
identically in $q$ \evP{}, so
$\omega=\omega_c\pm\sqrt{m^2+v^2k^2}$ at leading order with the
finite-$k$ composition law
$\cos(\omega-\omega_c)=\cos m\,\cos\varphi_{\rm massless}(k)$,
$\varphi$ the massless branch phase---accurate to $8\times10^{-4}$ at
$|k|=0.15$, the residual being the inter-sector spinor mismatch
\evM{}.
No real Dirac mass exists at four components: the anticommutant of the
quaternionic triple $\{C_a\}$ in $M_4(\mathbb R)$ is zero-dimensional
\evM{} (for a $\mathrm{Cl}(3,0)$ triple with squares $+1$ the
anticommutant is two-dimensional with strictly negative squares---
wrong-sign masses only); eight components are minimal, as realized here.

The census of Sec.~\ref{sec:census} has a sharp massive fate, charted
by the same sweep at $q\in\{0.08,0.15\}$, $q_m=0.05$ \evM{}. The mass
layer breaks all axis permutations (exactly the 8 diagonal sign flips
survive as spectral symmetries) and makes the spectrum self-conjugate
at every $k$, so both quasienergy branches live at the same momenta.
The \emph{resonant principal quartet}---the corner and the three
half-points, the self-conjugate nodes with $2k\equiv0$ (mod $\pi$),
where $\omega\equiv-\omega$---is gapped by exactly
$2\arctan[q_m/(1-q_m)]$ (two exactly two-fold levels split by $2m$ to
$10^{-15}$). \emph{Every one of the 40 extra Weyl points survives
exactly ungapped}: for each, the original doublet at $\omega$ and its
inversion partner at $-\omega$ re-close at a common point (coinciding
to $6\times10^{-14}$, displaced below $0.005\pi$ from the massless
location) with unchanged integer charges. The mechanism is
off-resonance: inversion doubling places the partner at the conjugate
quasienergy $-\omega\neq\omega$, and the mass layer gaps only resonant
sector pairs. Survivor charges sum to zero per branch. The massive
model also adds structures with no massless counterpart: degeneracy
curves confined to the $\{k_a\in\{0,\pi/2\}\}$ mirror planes (where the
inter-sector mass matrix element vanishes by symmetry), isolated Weyl
nodes pinned at the resonant quasienergies $\omega\in\{0,\pi\}$ (16
events at $q=0.08$, charges summing to zero), and a 12-line edge
network that remains exactly degenerate, with the original line
doublet split only weakly ($2m/9.8$ at $(0.7,\pi/2,\pi/2)$, $q=0.08$),
leaving sub-gap remnants. The massive model is therefore a Dirac
quartet embedded in a spectator population that stays massless.

Inversion doubling also cancels the chiral anisotropy of
Sec.~\ref{sec:chiralaniso}: the term is odd under the
inversion that relates the two sectors, and the Dirac branches are
symmetric functions of the sector pair, so on the doubled model every
single-axis sign flip leaves the branch dispersion invariant to
machine precision---below $10^{-15}$ at every radius to $|k|=1.2$ and
every mass density tested, against a relative splitting difference
rising through $13\%$ at $|k|=0.02$ and $41\%$ at $|k|=0.08$ for the
undoubled model; the transposition-chiral remnant is
$8\times10^{-4}$ at $|\delta k|=0.08$ \evM{}. What survives is even in
$k$: the family-ordered curvature spread of
Fig.~\ref{fig:mass} is the even remnant of the same term, exact for
the automaton; its magnitude sets the remaining distance to a
Lorentz-symmetric dispersion (open problem (iii),
Sec.~\ref{sec:discussion}).

\begin{figure*}[!t]
\includegraphics[width=0.85\textwidth]{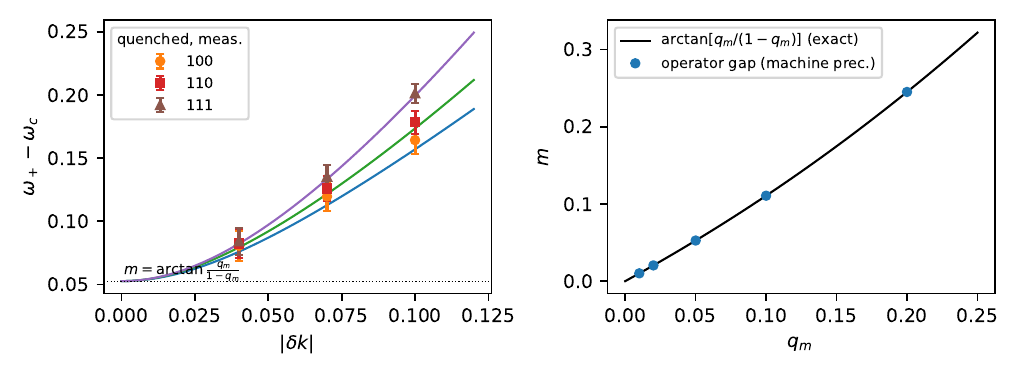}
\caption{Left: measured quenched massive dispersion (markers,
$q=0.08,\ q_m=0.05$) on the exact operator branches (lines); the family
ordering of the curvatures is the exact model's. Right: the operator gap
against the closed form, Eq.~(\ref{eq:massgap}), agreeing to machine
precision. \evX{}, \evP{}}
\label{fig:mass}
\end{figure*}

Measured on the quenched vacuum with the gated eight-launch instrument
(Fig.~\ref{fig:mass}): the gate row returns $m=0.0522\pm0.0056$ on the
exactly known operator (closed form $0.0526$; the replication's gate
row returns $0.0569\pm0.0026$ on the same operator, so the half-range
estimator's upward bias is of the size of the statistical error and
the mass central values are bias limited), and the quenched row
returns
$m=0.0569\pm0.0040$ ($+1.1\sigma$; the gap estimator is a half-range
over noise-split pole phases and is upward biased under noise, to
which part of this excursion is attributable; an independent-seed
replication at doubled $R$ returns $0.0557\pm0.0041$, $+0.8\sigma$)
with multiplet center
$\omega_c=0.3224\pm0.0103$---consistent with the closed forms within
statistics, as Theorem~\ref{thm:freshtape} extended to the massive model
demands---and the dispersion lies on the exact operator branches with
the family ordering of the curvatures that the operator prescribes
\evX{}.

\section{Interaction}
\label{sec:interaction}

For fixed media the multi-carrier sector is exactly determinantal
(Sec.~\ref{sec:coincube}), so genuine interaction requires
\emph{back-reaction}. The imprint layer L4 acts once per cycle: at sites
where an autonomous four-valued field fires (pair $(xy)$, $(yz)$ or $(zx)$
with probability $g/3$ each) \emph{and} the site's carrier fermion parity
is odd, the selected two environment species at that site are swapped.
The imprint field is autonomous in the same sense as the media: its
randomness lives in the initial product measure only (value $0$ with
probability $1-g$; the four values are carried as two bits), and it
evolves by legal deterministic pair-swap streaming at three
phase-continuing swaps per cycle along a fixed axis---a once-per-cycle
field outruns the carrier iff its speed is at least $3$ (a rotating-axis
schedule would need at least $7$)---lifting with the standard reordering
parity like every environment layer. A per-cycle resampled field
would be a stochastic rule---an annealed approximation, not an
automaton---and is not used.
Three design facts govern the layer. (i)~A fixed swap \emph{pair}
breaks cubic symmetry at $O(g)$ (anisotropic node corrections at
$5.9\sigma$ and $12\sigma$ in the two independent channels) \evX{}; the
rotating triple is protected. In the model defined here the protection
is complete and needs no symmetry argument: the interaction law,
Eq.~\eqref{eq:intlaw} below, makes the interacting cycle operator a
\emph{scalar} multiple of the free one at every cycle, so every
spectral property---the isotropy of the leading cone and the anisotropy
structure of Sec.~\ref{sec:chiralaniso} alike---is independent of $g$
\evP{}. (Under generic schedules, where the law holds only at leading
order, full quenched $A_4$ covariance is \emph{not} established: the
double axis reversals are verified quenched per mode \emph{including}
the imprint back-reaction \evM{}, cyclic axis rotations are verified at
the annealed-operator level \evM{}, single-axis reversals are provably
not symmetries---no signed relabeling exists among all 384---and the
measured interacting node isotropy under the generic schedule is
constrained only at the few-percent level. None of this is needed for
the model of this paper.)
(ii)~Parity control (not occupancy) is particle--hole even, so the complex
structure of Sec.~\ref{sec:exactqm} survives at $g>0$ unchanged \evM{}.
(iii)~The lift gauge of the imprint is \emph{physical}: the permutation
lift ($|11\rangle\to-|11\rangle$) and the Givens lift differ by a
parity-controlled sign with identical classical dynamics but per-event
coherent deficits $2q^2$ versus $2q(1-q)$; we fix the permutation lift
(minimal decoherence) \evM{}.

Certificates. The connected two-particle correlator
$C_2=G_2-\det G_1$ is nonzero already at $g=0$---the model is an
interacting carrier--environment theory at the correlator level, with
free-fermion structure only per medium realization---so the interaction
certificate is differential: $\Delta C_2(g)=C_2(g)-C_2(0)$ vanishes
identically at $g=0$ and switches on linearly,
$\Delta C_2/g\simeq0.35$ on the exact substrate (a one-dimensional ring
with a fixed swap pair and a static imprint field---an interaction
certificate, not the production three-dimensional layer) \evM{}. A control-variant
experiment (occupancy-controlled imprint, used only as an instrument
control) separates the parity-blocking contact artifact: $2$--$3\%$ of the
signal \evM{}.

Write events and the flush. At $g>0$ the imprint swaps two environment
bits at the carrier's site---a \emph{write}---and its permutation-lift
sign reads both bit values into the amplitude. Freshness must therefore
cover writes. The read$\to$write direction is closed by the streaming
alone (between any conversion read and a later write the bit has
outrun the carrier: writes only ever touch never-read bits, so every
lift sign-read is fresh), but a written or sign-read bit could be
consulted at the very next read. The model therefore includes one
post-imprint \emph{flush} batch per environment field each cycle (the
same certified layer type, origins continuing the field's phase), after
which every written bit moves at least three sites before any field is
read again: all reads are fresh at all times, and each cycle's imprint
contributes the exact scalar $(1-g)+g(1-2q^2)$, independently of
everything else. The parity control makes the vacuum
imprint-transparent, and
\begin{equation}
U_g(k)=(1-2gq^2)\,U(k)\qquad\text{exactly, at every cycle}
\label{eq:intlaw}
\end{equation}
\evP{}: the cone, residues and mass are untouched, at a damping cost
$\Gamma_{\rm int}=2q^2g$ per cycle ($\sim0.3\%$ of the free width at
$q=0.08$, $g=0.1$). Machine check, exact at $T=3$: all $2^{18}$
conversion histories $\times\,4^3$ imprint outcomes ($1.7\times10^7$
leaves), bits tracked through the streaming \emph{and} through the
write swaps, lift signs integrated exactly---
$G_{\rm quenched}(t)=(1-2gq^2)^t\,G_{\rm annealed}(t)$ endpoint by
endpoint to $2.1\times10^{-15}$, with zero correlated-read events of any
kind; without the flush the law is exact at the first cycle but fails
from the second through the identified sign-read channel (maximum
endpoint deviation $1.1\times10^{-3}$), which is why the flush is part
of the model \evM{}. One finite-size scope applies: with the flush each
environment field advances nine sites per cycle---three batches, two
from its own axis substeps and one from the flush---against a carrier
bounded by two sites per cycle along the streaming axis, so the
free-model horizon $T\le\lceil L/8\rceil$ does not transfer. The
closing speed $11$ gives $L>11T+6$ as a sufficient condition, and
enumerating the wrap condition on the real substep word sharpens it:
for $T\le12$ no wrap re-read exists once $L>11T-8$, the largest
wrap-admitting size being $22$ at $T=3$ \evM{}. Both
computations that invoke the law on a torus clear it with margin (the
exact enumeration at $L=48$ and the Monte Carlo run at $L=40$, both
$T=3$). The identity that anchors the measurement: the
media-ensemble walker average \emph{equals} the canonical
coherent-vacuum correlator exactly \evM{}. At Monte Carlo scale, paired
common-noise walkers (exact $g=0$ gate: the branches agree bit for bit)
confirm the law at every retained cycle within the flip-shot-noise
tolerance, with momentum spread at the noise floor ($\le4\times10^{-4}$
wherever all probe momenta are retained; at the last cycle the
amplitude gate retains a single probe momentum). At these parameters
the shot-noise tolerance exceeds the law's total deviation from unity,
so the run's measured content is the isotropy and the per-cycle
consistency; the law's \emph{magnitude} rests on the exact
enumeration above---isotropy survives interaction exactly, by theorem
and by measurement \evX{}.

\section{Damping: a theorem and its unitary completion}
\label{sec:damping}

\begin{figure}[!htb]
\includegraphics[width=0.75\columnwidth]{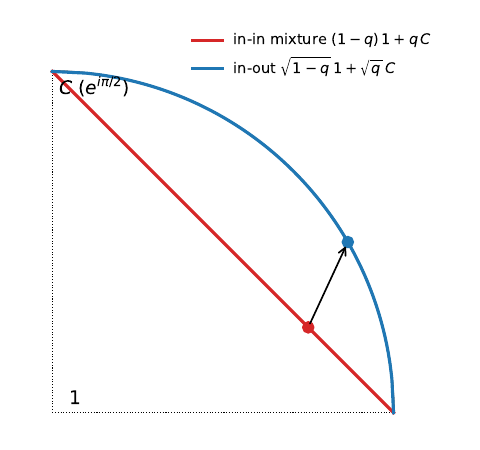}
\caption{Geometry of Theorem~\ref{thm:qbound}. Probabilistic (in-state)
mixtures of the identity and a quaternion unit lie on the chord (damped
interior); the two-boundary weights lie on the unitary arc. The isotropic
node splitting is a property of the whole family
$E_a(\alpha\openone+\beta C_a)$, for all weights.}
\label{fig:chord}
\end{figure}

The dressed quasiparticle of the classical in-state observable is damped:
$\Gamma(q)=-3\ln[(1-q)^2+q^2]$ exactly, with quality factor
$Q=\omega_0/\Gamma$ between $0.60$ and $0.90$ across the working range,
\emph{not} improving as $q\to0$ inside the relativistic window (both
$\omega_0$ and $\Gamma$ scale with $q$), and bounded at every density
by Proposition~\ref{prop:qglobal} below. Proposition~\ref{prop:factor}
states precisely what kind of damping this is: $U(k)=\rho^6V(k)$ with
$V(k)$ unitary, so $\Gamma=-3\ln\rho^2$ is a
mode-blind, $k$-independent \emph{interference-visibility} decay---every
coherent mode decays at the same rate relative to the conserved
probability sector, whose Perron eigenvalue in the doubled description
remains exactly $1$. It is not a broadening of the node relative to any
other mode. Nor is it a removable normalization: the in-state
normalization is fixed by probability conservation, and the
visibility-to-probability ratio is physical. In the fresh-tape model
there is no mode-selective decoherence at all
(Theorem~\ref{thm:freshtape}); under generic schedules it resides in
the re-read corrections (Sec.~\ref{sec:generic}). The overdamping is
global, not a property of the working density:

\begin{proposition}[Global overdamping]\label{prop:qglobal}
At $k=0$ each sub-step factor is $\rho$ times a unit quaternion of
angle $\theta(q)=\arctan[q/(1-q)]$, and the cycle is $\rho^6$ times
left multiplication by their product; the bi-invariant triangle
inequality for rotation angles on unit quaternions gives
$\omega_0\le6\,\theta(q)$, hence
\begin{equation}
Q\;=\;\frac{\omega_0}{\Gamma}\;\le\;\frac{2\,\theta(q)}{-\ln\rho^2(q)}
\;\le\;\frac{\pi}{2\ln2}\;=\;\kappa
\label{eq:qglobal}
\end{equation}
at every density, the last step because the middle bound is increasing
in $q$ with value $\kappa$ at $q=\tfrac12$. \evP{} The model's actual
supremum is sharper: $Q(q)$ increases monotonically with
$\sup_qQ=\pi/(3\ln2)=1.511$ approached as $q\to\tfrac12$ \evM{}, so
the in-state excitation retains at most $e^{-2\pi/Q}\le2^{-6}$ of its
visibility per oscillation period, at every density
($Q(0.08)=0.66$).
\end{proposition}

That the visibility
decay cannot be tuned away is structural:

\begin{theorem}[Q-pinning, quaternionic event class]\label{thm:qbound}
Let every conversion event lie in the quaternionic set
$\mathbb R\cdot Q_8=\{\pm\openone,\pm C_x,\pm C_y,\pm C_z\}$. Then for
any product of convex mixtures of events---non-commuting factors and
arbitrarily correlated weights included---every eigenvalue $\lambda$ of
the composite obeys
$\mathrm{dist}(\arg\lambda,(\pi/2)\mathbb Z)\le\kappa\,(-\ln|\lambda|)$
with $\kappa=\pi/(2\ln2)$ exactly. Moreover, under three further
hypotheses---(i) events in $\mathbb R\cdot Q_8$ controlled by
per-engagement reads of binary media; (ii) the generic-schedule read
geometry, in which the two same-axis engagements read lattice-adjacent
sites with both offsets realized across coin branches; (iii) a
deterministic
environment rule with randomness confined to the initial measure---an
isotropic node forces $\Gamma>0$: node unitarity requires an
almost-surely deterministic conversion word; among media factorizing
over axis lines, the media with deterministic read-window counts are
exactly the per-line crystal products, and every member destroys the
cone. \evP{}
\end{theorem}

(For the fresh-tape coincube itself no such hypotheses are needed:
$\Gamma=-3\ln\rho^2>0$ holds exactly by Proposition~\ref{prop:factor}
and Theorem~\ref{thm:freshtape}; the second statement addresses whether
\emph{any} read geometry and medium in the class could do better, and
answers no within its hypotheses.)

The restriction to the quaternionic class is necessary: for
transposition events on three or more modes (eigenvalue angles
$\{0,\pi\}$) products of mixtures develop three-cycle eigenvalues near
$e^{2\pi i/3}$, and the phase-per-decay ratio is unbounded---for the
product $[(1-\epsilon)S_{12}+\epsilon\openone]
[(1-\epsilon)S_{23}+\epsilon\openone]$ of two transposition mixtures it
is $4.95$ at $\epsilon=0.05$ and $131$ at $\epsilon=0.002$; for event
sets with finest eigenvalue angle $2\pi/\ell$ the constant grows as
$\kappa(\ell)\to2\ell/\pi$ \evM{}. The bound is a statement about the
quaternionic class, degrading linearly with coin refinement. In these
terms the theorem reads: no legal in-state automaton of this class can
hold its node quasienergy farther than $\kappa\Gamma$ from
$(\pi/2)\mathbb Z$. The bound governs the node---momenta where the
transport phases vanish, so the composite stays inside the group
algebra; at generic $k$ the momentum factors $E_a(k)$ leave the event
class and the ratio exceeds $\kappa$ already on the coincube operator
(reaching $6.06$ at $q=0.02$, $k=0.4\,e_x$) \evM{}. And ``an isotropic
node forces $\Gamma>0$'' asserts a nonzero uniform visibility decay,
not a node width.

The media hypothesis of the second statement is likewise scoped: for
general three-dimensional translation-invariant media the
deterministic-count classification is open---a family staggered along
the read axis with an independent phase per transverse line lies
outside the crystal list. It does not rescue the cone. Cross-streaming
makes a carrier's two same-axis reads sample \emph{distinct} transverse
lines (adjacent under the generic schedule; the argument needs only
distinctness), so per-line determinism does not give ensemble
determinism: for
independent phases of density $r$ the ensemble one-cycle operator is
damped by exactly $\Gamma=-3\ln[r^2+(1-r)^2]$ per cycle, vanishing only
at the crystal limits $r\in\{0,1\}$; and an exhaustive supercell
construction (all 4096 members of the family on the $2^3$ supercell,
each exactly unitary and closed under the model's streaming) finds
every one of its 11\,864 degenerate band groups failing an
isotropic-pair witness at the folded corner momenta---a witness that
does detect the unitary-arc cone when fed one. Interior supercell
momenta are not exhaustively excluded. \evM{}

The idealized ordered-dimer environment does achieve the commutator-limited
decoherence $\Gamma=k^2/2$ per event---the branch-remerging mechanism is
real---but its velocity fan collapses to the diamond: within the legal
class, cone and in-state unitarity are mutually exclusive. Within the
quaternionic class,
\emph{an isotropic cone forces overdamping}---$\Gamma>0$ with
$Q\le\pi/(2\ln2)$ at every density (Proposition~\ref{prop:qglobal}):
the excitation always retains less than $2^{-6}$ of its visibility per
oscillation period, an $e$-fold lost per tenth of an oscillation at
the working density---\emph{and in-state unitarity forces factorized
transport}; both halves are proven above. The postselected two-boundary
family below is the dichotomy's only known evasion; under the fresh
tape it is exact at all times.

The evasion is the two-boundary object---a \emph{postselected}
amplitude, fixed by a uniform final environment boundary in the
transfer-matrix formalism, not the evolution of a prepared state;
whether the induced two-boundary weight is non-negative, as a
probabilistic-automaton observable requires, is not addressed here and
joins the open problems. The isotropic node splitting
holds for the whole family $E_a(\alpha\openone+\beta C_a)$, and by the
identity of Proposition~\ref{prop:factor},
$T^\dagger T=(\alpha^2+\beta^2)\openone$ for every real weight pair (the
two-boundary walk is the $\alpha^2+\beta^2=1$ member of the same
family): \emph{every} product final boundary
$(w_0,w_1)$ yields zero relative damping and an isotropic node, landing on
the unitary arc at $s^2=w_1^2q/[w_0^2(1-q)+w_1^2q]$ \evP{}. The
completions thus form a one-parameter family; in-state weights $(1-q,q)$
lie on the damped chord, and the square-root point
$(\sqrt{1-q},\sqrt q)$ (Fig.~\ref{fig:chord}) is selected uniquely as
the only \emph{medium-blind} member: among product final boundaries,
$w_0=w_1$ is the unique choice invariant under the medium's bit-flip
involution---a criterion on the boundary weights alone---and the
$\sqrt q$ structure of the induced transfer pair is its consequence,
not its definition \evP{}. Numerically: all
$|\lambda|=1$ to $4\times10^{-16}$ with slope isotropy asserted below
$10^{-3}$ and $\omega_0$ tunable at the $\sqrt q$ scale \evM{}. In
the fresh-tape model the two-boundary amplitude inherits the
no-re-read property of Theorem~\ref{thm:freshtape} (the same path sum,
different boundary weights), so the in--out family is exactly unitary
at \emph{all} times---machine check: the exact quenched in--out path
sum under the fresh schedule equals the arc-operator power to
$1.6\times10^{-15}$ at every cycle ($L=24$, $T=3$, three momenta)
\evM{}; under generic schedules the fresh-read law is
exact at the first cycle only, with re-read corrections growing to
$11$--$24\%$ by three cycles at $q=0.08$. What remains open is not the
survival of the object but its interpretation: the positivity of the
induced two-boundary weight (Sec.~\ref{sec:discussion}).
\evP{}, \evM{}

\section{Discussion}
\label{sec:discussion}

Under the strictest single-fermion contract, the construction half of
the three-dimensional problem posed in Ref.~\cite{Wetterich2022} is
solved exactly: a legal probabilistic cellular automaton whose dynamics
is exactly unitary quantum mechanics, whose quenched single-particle
propagator \emph{is} its Bloch operator theory
(Theorem~\ref{thm:freshtape}), and whose single-carrier excitation on a
clustering product vacuum is a Weyl fermion with an exactly isotropic
leading cone---dispersion and spinor residues, exact by theorem and
confirmed by gated measurement---extendable to a massive Dirac quartet
with an exact tunable mass, and to a certified media-mediated
interaction whose damping law is exact at every cycle. The
Lorentz-symmetric half of the problem is delivered at leading order
only: beyond it stands the structural chiral anisotropy of
Sec.~\ref{sec:chiralaniso}, whose mechanism, exhaustive
unavoidability, and partial remedies this paper charts. Every closed form is a statement
about the automaton itself, every measurement runs behind a
known-answer gate, and the estimator's anisotropy resolution is
established by positive controls. The solution is a lattice fermion
with a computed spectator census: the sweep of Sec.~\ref{sec:census}
places 43 further Weyl points and three nodal lines at momentum
distance $\ge0.41\pi$ from the working node, and the massive model gaps
only the resonant quartet among them.

The precise open problems are: (i) the half-filled carrier sector of the
complex-structure construction in three dimensions; (ii) the positivity of the
induced two-boundary weight (its long-time survival is settled by
Theorem~\ref{thm:freshtape}; re-read resummation remains relevant only
for generic schedules);
(iii) the Lorentz-symmetric continuum limit: the chiral $O(|k|)$
anisotropy is structural at four components
(Sec.~\ref{sec:chiralaniso}), cancels at first order only on the
inversion-doubled branch, and is reduced but not removed by the legal
counterterm family ($45\%$) and by density optimization
($c\propto1/q$)---an exact cancellation mechanism, or an enlarged
construction class in which the term is not tied to the net chirality, is the open
problem, interaction corrections beyond the scalar law included; and
(iv) multi-carrier physics beyond the contact terms already present in the
extracted quartic action; and (v) the spectator population of the
massive census---the exactly massless off-resonant Weyl species, the
pinned nodes, and the sub-gap nodal-line remnants---and whether a legal
layer can gap it: the octet pair persists with constant opposite
chiralities across the entire density range, separated in both momentum
and quasienergy, so no static translation-invariant perturbation
suffices except near $q\approx0.339$ and $0.370$, where
opposite-chirality quasienergy crossings open momentum-only bridges
\evM{}. The damping of the in-state excitation is
not on this list: it is an exact, mode-blind prediction of the
construction (Proposition~\ref{prop:factor}) of stated
magnitude---$Q\approx0.66$ at the working density, $Q\le\pi/(2\ln2)$ at
every density---and the dichotomy, cone forces overdamping
(Proposition~\ref{prop:qglobal}) while in-state unitarity forces the
diamond (Theorem~\ref{thm:qbound}), stands as the
second principal result of this paper, with the postselected
two-boundary family as its only known evasion, exact at all times under
the fresh tape.

\begin{acknowledgments}
Derivations, code, verification protocols and drafts were prepared with
substantial assistance from the AI system Claude (Anthropic); all results
were machine-verified as described, and the author takes full
responsibility for the content. The complete verification code and the
manuscript source are available at Ref.~\cite{repo}.
\end{acknowledgments}

\appendix

\section{Proofs}
\label{app:proofs}

\subsection{Theorem \ref{thm:finiteray} (finite rays)}

Let the rule conserve particle number in the one-particle sector. Its
restriction to that sector is a signed permutation of the species basis
composed with site shifts, so the Bloch operator has matrix elements
$S(k)_{\alpha\beta}=s_{\alpha\beta}e^{ik\cdot d_{\alpha\beta}}$ with, by
the unique-jump property, exactly one nonzero entry in each row and column:
a monomial matrix. Its characteristic polynomial factorizes over the cycles
of the underlying permutation; a cycle of length $\ell$, accumulated
displacement $D$ and sign $\sigma$ contributes
$\lambda^\ell=\sigma\,e^{ik\cdot D}$, i.e.
$\omega_m(k)=(k\cdot D+\arg\sigma+2\pi m)/\ell$: every branch is exactly
linear with group velocity $D/\ell$, drawn from the finite set of cycle
data. A cone $\omega=v|k|$ requires a two-sphere of group velocities. The
deviation of any finite ray set from a sphere is scale invariant, so no
continuum limit removes it; and a periodic cycle of rotated rules is again
unique-jump (products of signed permutations are signed permutations), so
the theorem applies to the composite rule. \hfill$\square$

\subsection{Theorem \ref{thm:ntwo} and the accompanying search}

Three ingredients; the theorem is (ii), the rest is the search.
(i) \emph{Algebraic ceiling.} On
$\mathfrak{sl}_2(\mathbb R)$ every element satisfies
$X^2=-\det(X)\openone$ and anticommuting elements are orthogonal in the
polarization of $\det$; in the basis $\{X,Z,XZ\}$ the squares are
$+\openone,+\openone,-\openone$, so the form has signature $(2,1)$ and no
real two-dimensional representation of $\mathrm{Cl}(3,0)$ exists.
(ii) \emph{Generic weights.} A symbolic orthogonality lemma for the
mixture families shows that a degenerate propagating pair at any momentum
requires $(1-p)^2=p^2$ for rotation coins, is impossible for reflection
coins, and degenerates for diagonal coins: a linear node forces
$U(k_0)=\lambda_0\openone$, the splitting generators are conjugated
involutions with automatically scalar anticommutators ($2\times2$
Cayley--Hamilton), and the required bilinear orthogonality gives
$T_{00}T_{11}=-T_{01}T_{10}$. For the blocked two-origin cycles the
exclusion is exhaustive in the tables and heuristic in momentum: all
512 sign/direction designs per (placement, engagement, $p$) sector are
enumerated, with $p\in\{0.15,0.25,0.35,0.5\}$ and 14-start continuous
momentum solving, and every design carries an explicit certificate---a
candidate node entering the sector's anisotropy floor, a
degenerate-only scalar point, or no scalar point at all, the last
certified by a minimum scalar residual $\ge0.074$ against the $10^{-9}$
solver tolerance of solvable designs. The floors are claimed at the
sampled grid values of $p$ only; between them they dip toward zero on
approach to the tuned $p^*$, which is the content of part (iii), not a
violation.
(iii) \emph{The fine-tuned points.} At the singular weight $p=\tfrac12$
the single-origin placement has $U(k_0)$ scalar at
$k_0=(\tfrac{3\pi}4,\tfrac{3\pi}4,-\tfrac\pi4)$ (and its symmetry
images), and the three effective generators form an exact complex
Clifford triple, giving $\omega=\pi\pm|\delta k|$: an isotropic node with
every parameter frozen and $\Gamma=\tfrac32\ln2$ at the chord midpoint.
In the blocked placements, branch continuation in $p$ (warm-started exact
scalar solves with probe-radius-extrapolated anisotropy) locates
codimension-one transversal zeros of the single remaining Gram condition
at $p^*=0.3300$ (additive placement,
$k_0/\pi=\pm(0.916,0.916,-0.916)$, $|\lambda_0|=0.1736$, $\chi=-1$) and
$p^*=0.3204$ (multiplicative, $|\lambda_0|=0.1799$, $\chi=+1$), both
with $\omega_0=\pi$: isolated tuned nodes, not parameter families. All parts are machine-checked, the sweep exhaustively.
\hfill$\square$

\subsection{Theorem \ref{thm:freshtape} (fresh tape)}

\emph{Lemma A.} A pair-swap layer with origin $o$ moves the bit at 1D
site $y$ to $y+1$ if $y\equiv o\ (\mathrm{mod}\ 2)$, else to $y-1$
(finite check, four cases).

\emph{Lemma B (monotone transport).} Under any origin-alternating swap
sequence with phase $p$ (the $k$-th swap has origin $(p+k)\bmod2$), the
class $\epsilon=(-1)^{y+p}$ of a bit is conserved and its displacement
after $k$ swaps is exactly $\epsilon k$: by Lemma A the bit moves $+1$
iff its position parity matches the current origin, and both flip at
each step, so the direction is permanent. (A schedule restarting each
batch at origin $0$ breaks the alternation at the batch boundary; its
transport is $0$ per batch---the machine-checked negative control.)

\emph{Lemma C (slot counting).} Relative to the axis-$a$ field the
cycle is the periodic word $AASSXX$ ($A$: axis-$a$ substep---the path
reads the field at its own site, then the field streams one batch; $S$:
axis-$s_a$ substep; $X$: third axis). Between two $A$-slots separated
by $\Delta n$ batches the path takes at most $m\le\Delta n+1$
$s_a$-steps, and when $\Delta n$ is even, $m=\Delta n$ (exhaustive over
12 cycles).

\emph{Proof of (a).} Pair-swap layers permute the axis-$a$ field only
along $s_a$, so a re-read at substeps $\tau_1<\tau_2$ requires zero net
path displacement along both transverse axes and equal path and bit
displacements along $s_a$. By Lemma B the bit moves exactly $3\Delta n$
sites in a fixed direction; by Lemma C the path moves at most
$\Delta n+1$. Since $3\Delta n>\Delta n+1$ for every $\Delta n\ge1$,
the match fails at every separation---no residual window. As an
independent check that the reduction models the automaton faithfully,
the exact path enumerations (labels tracked through the real
permutations, all signs and weights) cover every separation
$\Delta n\le7$ and find zero re-reads with
$G_{\rm quenched}=G_{\rm annealed}$ to $10^{-16}$.

\emph{Proof of (b) and sharpness.} On the ring of circumference $L$ the
match condition becomes $\epsilon\,3\Delta n\equiv\Delta y\ (\mathrm{mod}\ L)$
with $|\Delta y|\le m$ and $\Delta y\equiv m\ (\mathrm{mod}\ 2)$;
transverse return along axis $a$ forces $\Delta n$ even, and then
$m=\Delta n$ by Lemma C. The first solution is
$\Delta n^*=2\lceil L/8\rceil$, so an evolution of $T$ cycles is
re-read-free iff $2T-2<2\lceil L/8\rceil$, i.e.\ iff
$T\le\lceil L/8\rceil$. The certificate reproduces $\Delta n^*$ from
the real tapes at $L=12,16,24,32$, and the full path enumeration at
$L=16$ exhibits the first wrap re-reads at $T=3=\lceil16/8\rceil+1$
with minimal separation exactly $4$ batches. \hfill$\square$

\subsection{Particle--hole equivariance, $[\hat S,P]=0$}

\emph{Lemma (string factorization).} Let $B$ be a parity-even operator
with support $T$, and $T_c$ the carrier modes in $T$. The Majorana
factors of $P$ anticommute pairwise, so sorting gives
$P=\sigma\,P_{T_c}P_{T_c^\complement}$ with a fixed reordering sign
$\sigma=\pm1$; a parity-even $B$ commutes with every Majorana off its
support, hence with $P_{T_c^\complement}$, and
$[B,P]=\sigma\,[B,P_{T_c}]\,P_{T_c^\complement}$: the block-local check
suffices.

\emph{Layers.} L2 acts, per fired site, on the four same-site carrier
modes (plus the control); the stored Jordan--Wigner spectators lie
inside the contiguous block, so $[{\rm L2},P]=0$ reduces to one finite
block check per axis, independent of $L$. L1 is a canonical permutation
lift, for which $U_\pi PU_\pi^{-1}=\mathrm{sign}(\pi)\,P$; per channel
the axis-$a$ shift is $L^2$ disjoint $L$-cycles, so
$\mathrm{sign}(\pi)=(-1)^{(L-1)L^2}=+1$ for every $L$ (even $L$: $L^2$
even; odd $L$: $L-1$ even). L3 is supported on environment modes and
parity even, so it commutes with every carrier Majorana. L4's carrier
part is a parity projector and its environment part a string-free
adjacent swap, block-checked in both lift gauges. Size independence is
structural: the same finite blocks recur at every site of every
lattice, so the identity holds on the full Fock space at every $L$.
\hfill$\square$

\subsection{The mass identities, Eq.~(\ref{eq:massgap})}

First, the factorization
$U(k)=(\openone\otimes R_m)[U_4(k)\oplus U_4(-k)]$ with
$R_m=(1-q_m)\openone+q_mXZ$: (1)
$d^{(8)}=d\otimes\mathrm{diag}(1,-1)$, so every transport phase is
block diagonal in $b_m$ with opposite momentum sign in the two blocks,
$E_8(k)=E_4(k)\oplus E_4(-k)$; (2) $C_8=C_4\otimes\openone$, so
conversions are block diagonal and identical in the two blocks; hence
the pre-mass cycle is $U_4(k)\oplus U_4(-k)$ and the once-per-cycle
mass layer is $\openone_4\otimes R_m$ (machine-checked symbolically in
the full vector momentum). Second, the gap:
at $k=0$ the momentum factors are the identity, so the direction reversal
that distinguishes the two $b_m$ sectors is invisible and the axis layers
act as $U_4(0)\otimes\openone_2$. The mass layer acts on the $b_m$ factor
alone: $M_2=(1-q_m)\openone+q_m XZ$, a real multiple of a rotation with
$\arg$-eigenvalues $\pm\arctan[q_m/(1-q_m)]$ and modulus
$\sqrt{(1-q_m)^2+q_m^2}$. Hence
$U(0)=U_4(0)\otimes M_2$: the spectrum is the product set, frequencies add
angles, the multiplet center is unshifted, and the gap is exactly
$2\arctan[q_m/(1-q_m)]$, independent of $q$. \hfill$\square$

\subsection{The imprint amplitude law}

In the one-particle sector the carrier's site has odd parity, so the
imprint fires iff $\iota\neq0$, probability $g$. The permutation lift
of the fired pair swap has sign table $|00\rangle\to+|00\rangle$,
$|01\rangle\to+|10\rangle$, $|10\rangle\to+|01\rangle$,
$|11\rangle\to-|11\rangle$ (the canonical transposition lift), and the
swap preserves the Bernoulli pair measure, so the per-event coherent
factor is the plain sign average $(1-q)^2+2q(1-q)-q^2=1-2q^2$,
identical for the three pairs. One line of expectation:
$U_g=(1-g)U+g(1-2q^2)U=(1-2gq^2)U$ per cycle; under the fresh-tape
model every read is fresh, so the law holds exactly at every cycle. The
Givens
lift's table gives $1-2q(1-q)$ instead: the law is lift-gauge
dependent, the isotropy of the damping is not. \hfill$\square$

\subsection{Proposition \ref{prop:qglobal} (global overdamping)}

At $k=0$ the momentum factors are the identity and each sub-step
factor is $(1-q)\openone+qC_a=\rho\,[\cos\theta+\sin\theta\,C_a]$ with
$\theta=\arctan[q/(1-q)]$: $\rho$ times the left-regular representation
of a unit quaternion of rotation angle $\theta$. The cycle at $k=0$ is
therefore $\rho^6L_u$ with $u=u_z^2u_y^2u_x^2$ a product of six unit
quaternions of angle $\theta$ each, and its eigenphases are
$\pm\Theta(u)$, the angle of $u$. The bi-invariant metric on the unit
quaternions ($S^3\cong SU(2)$, geodesic distance = rotation angle)
satisfies the triangle inequality
$\Theta(uv)\le\Theta(u)+\Theta(v)$, so $\omega_0=\Theta\le6\theta$.
Dividing by $\Gamma=-3\ln\rho^2$ gives
$Q\le2\theta/(-\ln\rho^2)=:g(q)$; $g$ is increasing on
$(0,\tfrac12]$ (machine-checked densely) with
$g(\tfrac12)=(\pi/2)/\ln2=\kappa$, giving Eq.~\eqref{eq:qglobal}. The
model's supremum $\pi/(3\ln2)$ follows from
$\omega_0\to\pi$ and $\Gamma\to3\ln2$ as $q\to\tfrac12$, with $Q(q)$
monotone on a dense grid \evM{}; the per-period statement is
$e^{-2\pi/Q}\le e^{-2\pi\cdot3\ln2/\pi}=2^{-6}$. \hfill$\square$

\subsection{Theorem \ref{thm:qbound} (Q-pinning)}

Any product of convex mixtures of the event set is itself a convex mixture
over the quaternion group $Q_8$ (the product measure is the convolution on
the group, so correlations between events are automatically included). In
the quaternion representation such a mixture is
$T=w\openone+x C_x+y C_y+z C_z$ with $|w|+\lVert(x,y,z)\rVert_1\le1$,
and its eigenvalues are $w\pm i\lVert(x,y,z)\rVert_2$. Since
$\lVert\cdot\rVert_2\le\lVert\cdot\rVert_1$, every eigenvalue lies
inside the convex hull of the chords between adjacent fourth roots of
unity. On the extremal chord $1\to i$, $\lambda(t)=(1-t)+it$, the
maximized quantity is
$\mathrm{dist}(\arg\lambda,(\pi/2)\mathbb Z)/(-\ln|\lambda|)$: by the
square-hull reduction and the $\pi/2$ rotation symmetry of the event
angles it attains its maximum $\pi/(2\ln2)$ at the chord midpoint
$t=\tfrac12$, where the distance folds from $\arg\lambda$ to
$\pi/2-\arg\lambda$. (The unfolded ratio $\arg\lambda/(-\ln|\lambda|)$
diverges toward the chord endpoint and is maximized nowhere interior; a
dense hull scan confirms no interior point exceeds $\kappa$.) This bounds
$\mathrm{dist}(\arg\lambda,(\pi/2)\mathbb Z)\le\kappa\Gamma$ for
arbitrary, non-commuting products---the full statement of the theorem.
The group closure is the load-bearing hypothesis: outside
$\mathbb R\cdot Q_8$ (already for transposition events on three modes)
products of mixtures leave the span of the event set, three-cycle
eigenvalues appear, and the measured phase-per-decay ratios grow without
bound, as reported in the main text.
For the second statement (under hypotheses (i)--(iii), the read-window
structure being hypothesis (ii)): unitarity of the node eigenvalue
requires the conversion count $N$ (mod 4) along a read window to be
deterministic, by
$|\mathbb E[i^N]|^2=(p_0-p_2)^2+(p_1-p_3)^2\le1$ with equality only at a
vertex; on rings, and on three-dimensional media factorizing over axis
lines, the media with deterministic window counts are exactly the
(per-line) crystals (exhaustively enumerated at $L=6,7,8,10$), with the
non-factorizing phase-per-line family closed separately as described in
the main text; and the deterministic-word cycles' spectra are computed
explicitly and are never isotropic. Machine checks accompany every step,
including randomized non-commuting products. \hfill$\square$

\section{Verification protocols}
\label{app:verify}

Every claim tagged \evM{} is an assertion in a runnable script; the
tolerances quoted below are the assertion thresholds, and each script is
self-checking (it fails loudly rather than reporting).

\emph{Layer lifts.} The conversion layer is verified densely on the
one-site Fock space (32 states): the controlled double Givens is a signed
permutation of the basis, unitary, commutes with fermion parity, is the
identity on the empty control sector, has single-particle matrix equal to
the signed-swap block of $C_a$, satisfies $G^2=-\openone$ on the
one-particle sector, and maps the doubly occupied pair to $+$ itself
(determinant one), for all three axes, to $10^{-12}$.

\emph{Composite sign coherence.} The full cycle is built as a signed
permutation of the many-body basis from the stored per-layer sign rules
(spectator Jordan--Wigner strings for conversions; inversion and wrap
parity for translations), in the segregated mode ordering. For random
environment configurations the zero-, one- and two-particle sectors are
extracted and compared: the one-particle sector must equal the stated rule
and the two-particle sector the antisymmetrized square $\Lambda^2M_1$,
exactly. This holds on the full Fock space of a ring substrate ($2^{15}$
states) and in three-dimensional geometry at $L=2$ (496 pair states) and
$L=3$ (5778 pair states), under both the fresh-tape and the generic
streaming schedules. Mutation controls certify sensitivity: deleting
the spectator string produces $8/496$ mismatches; deleting translation
parity produces $2772/5778$ at $L=3$ (counts are specific to the
committed environment draw) and---the recorded caveat---$0$ at
$L=2$, where two sub-steps per axis make net crossings even.

\emph{Complex structure.} $P$ is applied as the Majorana string with its
Jordan--Wigner signs; $[\hat S,P]=0$ is checked as an identity of signed
permutations (both orderings byte-identical), on the full substrate Fock
space and, in three dimensions, on the sectors
$\{0,1,2,3,M-2,M-1,M\}$ (three-carrier wedges, $1.6\times10^6$
state--draw evaluations) using particle--hole duality to reduce the deep sectors to one-
and two-hole problems; per layer and for the full cycle, under both
streaming schedules.
The unitarity of the complex picture is verified as invariance of the
induced Hermitian form on random vectors to $10^{-9}$.

\emph{Grassmann extraction.} Local factors are built from the block tables
over exact rationals; $e^{-L}=K$ is asserted exactly before any action is
reported, and the inverse map $K\to(\text{perm},\text{signs})$ must
round-trip. The pipeline is validated term by term against the closed-form
interaction action of Ref.~\cite{Wetterich2022b}.

Table~\ref{tab:claims} maps each principal claim to its verification
script in Ref.~\cite{repo}.

\begin{table}[!htb]
\caption{Principal claims and their verification scripts (all
self-asserting; in \texttt{scripts/} of Ref.~\cite{repo} unless noted).}
\label{tab:claims}
\begin{ruledtabular}
\scriptsize
\begin{tabular}{ll}
finite-ray tests & \texttt{tests/test\_finite\_ray\_*}\\
$v(q)$ dressing law & \texttt{theory\_linear\_law}\\
Abelian-gauge diamond & \texttt{w2\_scalar\_gauges}\\
coin scans ($n=2,4$) & \texttt{w3\_transfer\_scan}\\
annealed cone & \texttt{w3\_cone\_verify}, \texttt{theory\_real\_cone}\\
lift certificates & \texttt{w3c\_lift\_check}\\
composite sign coherence (1D ring) & \texttt{w3c\_composite\_check}\\
3D sign coherence; sectors & \texttt{theory\_3d\_certs}\\
$[\hat S,P]$ proof blocks & \texttt{theory\_sp\_proof}\\
fresh-tape theorem & \texttt{freshtape\_proof}\\
re-read kinematics & \texttt{reread\_kinematics}\\
chiral anisotropy; no-go & \texttt{theory\_anisotropy}\\
global overdamping & \texttt{theory\_q\_global}\\
fresh-tape schedule tests & \texttt{tests/test\_freshtape}\\
quenched cone (gated) & \texttt{w3c\_fresh}\\
generic-schedule cone/helicity & \texttt{w3c\_corner}, \texttt{w4\_helicity}\\
generic mass/interaction & \texttt{m8\_corner}, \texttt{i3\_quenched3d}\\
census, charges & \texttt{census\_sweep}\\
factorization; structural loci & \texttt{spectrum\_census}\\
positive controls & \texttt{w3c\_positive\_control}\\
positive control (fresh) & \texttt{w3c\_control\_fresh}\\
two-component taxonomy & \texttt{theory\_redteam}\\
helicity residues & \texttt{w4\_helicity\_exact}, \texttt{w4\_fresh}\\
complex structure & \texttt{e1\_complex\_structure}\\
substrate bridge & \texttt{bridge\_check}\\
Grassmann actions & \texttt{e2\_coincube\_actions}\\
Eq.~(38) erratum & \texttt{wetterich\_eq38\_check}\\
mass sector & \texttt{m8\_mass\_exact}, \texttt{m8\_fresh}\\
mass structure & \texttt{theory\_mass}\\
mass-gap identity & \texttt{theory\_mass\_identity}\\
interaction exact law & \texttt{freshtape\_interaction}\\
interaction measured & \texttt{i3\_fresh}\\
interaction certificates & \texttt{i2\_connected}, \texttt{i2\_split}\\
imprint-field theory & \texttt{theory\_interaction}\\
spectator flow (open problem v) & \texttt{orbit\_flow}\\
rationality of closed forms & \texttt{theory\_rationality}\\
Q-pinning & \texttt{theory\_q\_bound}\\
in--out amplitudes & \texttt{a\_inout}, \texttt{a\_inout\_3d}\\
in--out under fresh tape & \texttt{inout\_fresh}
\end{tabular}
\end{ruledtabular}
\end{table}

\section{Measurement pipeline}
\label{app:instrument}

All \evX{} results use one instrument family. The object is the exact
per-medium single-particle signed field (the single-carrier sector is
linear in the field for fixed media, so this is walker-exact with infinite
walkers), averaged over $R$ media; the media-ensemble average provably
equals the canonical coherent-vacuum correlator.

\emph{Bloch-matrix estimation.} All $n$ basis launches are evolved, giving
the full $n\times n$ propagator series $G(t,k)$; the one-cycle operator is
fit by least squares, $\hat U(k)=G_{t+1}G_t^+$ over a post-transient
window. Momenta are probed off the lattice grid by direct phase
contraction, with probe radii inside the node's measured linear range.

\emph{Pooling.} Naive eigenvalue extraction near a two-fold node is
$\sqrt{\text{noise}}$-ill-conditioned (a noisy nearly defective matrix
splits its double eigenvalue as the square root of the perturbation);
pooling is therefore performed on characteristic-polynomial coefficients
over the symmetry orbit of measurement momenta---a basis-invariant,
\emph{linear} (hence unbiased) average---before rooting. Slopes use
symmetric splittings at two step sizes with $\delta k\to0$ Richardson
extrapolation. Media are accumulated in independent blocks (ten for the
cone instrument, eight for the massive instrument, six for the helicity
instrument) and every final quantity (slope, ratio) is jackknifed at the
end-of-pipeline level over leave-one-block-out samples; the annealed gate
row's deviation from the exact operator is recorded per channel and
reported alongside as the instrument systematic---the quoted errors are
the block jackknives. Off-symmetry direction families are measured alongside the
cubic stars. The diamond-exclusion significance is computed in code for
every channel by the generic-schedule instrument; for the fresh-model
rows the quoted significances are the committed ratios' distances from
the benchmark values in units of their own jackknife errors---direct
arithmetic on the stored rows---with the weakest channel quoted. The
three campaign instruments draw their media from a common base seed
stream, so their quenched excursions are correlated across
instruments; independent-seed replications at doubled $R$ accompany
the campaign: the cone replication reproduces the quenched node at
$-0.4\%$ of the closed form and the mass replication returns the gap
at $+0.8\sigma$.

\emph{Gates.} Every production run carries an annealed known-answer row
through the identical pipeline; a failed gate discards the run. The
cone instrument gates both rows against the closed forms; the massive
and helicity instruments gate the known-answer row only, so their
quenched rows are protected by the gate row's validation of the
pipeline rather than by a bound of their own. Amplitude
gates exclude momenta near free-beat amplitude nodes for ratio estimators
(shot-noise gating alone is insufficient there). Eigenvector work uses the
$e^{+ik\cdot x}$ contraction (Sec.~\ref{sec:spectrum}). For interaction
ratios, paired common-noise evolution (identical randomness with the
imprint on and off) cancels the medium-ensemble variance, which is
otherwise dominated by low-order conversion caustics. Gate tolerances are
set to catch catastrophic estimator failure, not to certify accuracy:
the cone-ratio gates use $\max(2\%, 3\sigma)$ of the row's own
jackknife, the fresh instruments gate both rows' node parameters at
$\max(0.5\%,3\sigma)$, the massive-gap gate allows $30\%$ of the gap
with the multiplet center gated at $0.02$, and the generic instrument's
node-modulus gate is $0.05$ absolute. Accuracy statements never rest on gates;
they rest on the quoted statistical errors and the recorded systematics,
and the estimator's ability to resolve anisotropy is established by
positive controls (a synthetic diamond-splitting truth and a legal
broken-coin walk) run through the identical pipeline. A committed
broken run under the generic schedule (probe radii outside the node's
linear window) documents why the controls carry the burden: despite a
$2.4\times$ absolute-speed failure it returns cubic ratios of $1$ at
the $10^{-4}$ level---a ratio near $1$ is, by itself, a weak
discriminator.

\emph{Recorded failure modes.} The following estimator pathologies were
each encountered, diagnosed and firewalled during this program, and are
listed because they generically afflict lattice pole measurements:
(i) $\sqrt{\text{noise}}$ splitting at near-degenerate poles;
(ii) exceptional-point misidentification in damped operators (real double
eigenvalues accepted as propagating nodes);
(iii) through-origin fits across a packet's death into the noise floor;
(iv) ratio blowup at beat-amplitude nodes;
(v) probing outside the (renormalized) linear cone window;
(vi) branch mixing under scalar single-pole estimators at two-fold
degenerate nodes;
(vii) one-dimensional co-streaming substrates, whose companion locking
invalidates absolute calibrations that do transfer in three dimensions.

\section{Closed forms}
\label{app:closedforms}

\begin{table}[!htb]
\caption{Exact expressions for the model's free and dressed spectra
(exact for the automaton by Theorem~\ref{thm:freshtape}). All are
verified against the operators they describe to at least
$10^{-9}$; $D=\sqrt{1-q}+\sqrt q$.}
\label{tab:closed}
\begin{ruledtabular}
\scriptsize
\begin{tabular}{ll}
node modulus & $|\lambda_0|^2=[(1-q)^2+q^2]^6$\\
factorization & $U=\rho^6V$, $\rho^2=(1-q)^2+q^2$, $V(k)$ unitary\\
free width & $\Gamma(q)=-3\ln[(1-q)^2+q^2]$\\
cone speed & $v(q)\colon\ 2/\sqrt3 \to 1.207$ (max at $q\simeq0.30$)\\
dressing law (1D) & $v=2|1-2q|$\\
mass gap & $2m=2\arctan[q_m/(1-q_m)]$\\
mass-layer modulus & $\sqrt{(1-q_m)^2+q_m^2}$ per cycle\\
massive dispersion & $\cos(\omega-\omega_c)=\cos m\,\cos\varphi(k)$\\
odd anisotropy & $s-2v|k|=c(q)\,k_xk_yk_z/|k|$, $c(0.08)=42.7$\\
interaction law & $U_g(k)=(1-2gq^2)\,U(k)$\\
Q-pinning constant & $\kappa=\pi/(2\ln2)=2.2662$\\
global $Q$ bound & $Q\le2\theta/(-\ln\rho^2)\le\kappa$;
$\sup_qQ=\pi/(3\ln2)$\\
in--out transfer & $E_a[\sqrt{1-q}\,\openone+\sqrt q\,C_a]/D$,
$|\lambda|\equiv1/D$
\end{tabular}
\end{ruledtabular}
\end{table}

Table~\ref{tab:closed} collects the closed forms. Two remarks. The
in--out transfer's uniform modulus $1/D$ per read is a global scalar
(every path reads one bit per sub-step) and is absorbed into wave-function
renormalization; the physical statement is the vanishing of all
\emph{relative} damping. The quality factor of the in-state quasiparticle,
$Q=\omega_0/\Gamma$, rises from $0.60$ to $0.90$ over
$q\in[0.02,0.25]$, monotonically to $\pi/(3\ln2)$ at $q\to\tfrac12$
(Proposition~\ref{prop:qglobal}); the quantity
Theorem~\ref{thm:qbound} bounds by
$\kappa$ is $\mathrm{dist}(\omega_0,(\pi/2)\mathbb Z)/\Gamma$, which
coincides with $Q$ for $q\lesssim0.2$ (where $\omega_0<\pi/4$).

\end{document}